\documentclass[11pt,letterpaper]{article}
\usepackage[margin=1in]{geometry}
\usepackage{amssymb,amsmath,amsthm}
\usepackage{graphicx}
\usepackage[colorlinks,citecolor=blue,linkcolor=blue,urlcolor=red,pagebackref]{hyperref}
\usepackage{subcaption}
\usepackage{booktabs}
\usepackage{lineno}
\usepackage{color}
\usepackage{enumitem}
\usepackage{tikz}
\usepackage[linesnumbered,lined,boxed]{algorithm2e}
\usepackage{framed}
\usepackage{thm-restate}
\usepackage[T1]{fontenc}

\newcommand*\patchAmsMathEnvironmentForLineno[1]{%
  \expandafter\let\csname old#1\expandafter\endcsname\csname #1\endcsname
  \expandafter\let\csname oldend#1\expandafter\endcsname\csname end#1\endcsname
  \renewenvironment{#1}%
     {\linenomath\csname old#1\endcsname}%
     {\csname oldend#1\endcsname\endlinenomath}}%
\newcommand*\patchBothAmsMathEnvironmentsForLineno[1]{%
  \patchAmsMathEnvironmentForLineno{#1}%
  \patchAmsMathEnvironmentForLineno{#1*}}%
\AtBeginDocument{%
\patchBothAmsMathEnvironmentsForLineno{equation}%
\patchBothAmsMathEnvironmentsForLineno{align}%
\patchBothAmsMathEnvironmentsForLineno{flalign}%
\patchBothAmsMathEnvironmentsForLineno{alignat}%
\patchBothAmsMathEnvironmentsForLineno{gather}%
\patchBothAmsMathEnvironmentsForLineno{multline}%
}

\newtheorem{theorem}{Theorem}
\newtheorem{definition}[theorem]{Definition}
\newtheorem{corollary}[theorem]{Corollary}
\newtheorem{claim}[theorem]{\bfseries{Claim}}
\newtheorem{lemma}[theorem]{\bfseries{Lemma}}

\newcommand{\eps}{\varepsilon}

\renewcommand{\leq}{\leqslant}
\renewcommand{\geq}{\geqslant}

\newbox\ProofSym
\setbox\ProofSym=\hbox{%
\unitlength=0.18ex%
\begin{picture}(10,10)
\put(0,0){\framebox(9,9){}}
\put(0,3){\framebox(6,6){}}
\end{picture}}

\begin{document}

\title{
  A Nearly Quadratic-Time FPTAS for Knapsack
}

\author{
  Lin Chen\thanks{chenlin198662@zju.edu.cn. Zhejiang University. Part of this work was done when the author was affiliated with Texas Tech University.}
  \and 
  Jiayi Lian\thanks{jiayilian@zju.edu.cn. Zhejiang University. Part of this work was done during an internship at Ant Group. Supported by Ant Group through CCF-Ant research fund [Project No. RF20220212]}
  \and
  Yuchen Mao\thanks{maoyc@zju.edu.cn. Zhejiang University. Supported by National Natural Science Foundation of China [Project No. 12271477]}
  \and
  Guochuan Zhang\thanks{zgc@zju.edu.cn. Zhejiang University. Supported by National Natural Science Foundation of China [Project No. 12131003]}
}

\date{}

\maketitle

\begin{abstract}
    We investigate the classic Knapsack problem and propose a fully polynomial-time approximation scheme (FPTAS) that runs in $\widetilde{O}(n + (1/\varepsilon)^2)$ time. This improves upon the $\widetilde{O}(n + (1/\varepsilon)^{11/5})$-time algorithm by Deng, Jin, and Mao [\textit{Proceedings of the 2023 Annual ACM-SIAM Symposium on Discrete Algorithms, 2023}]. Our algorithm is the best possible (up to a polylogarithmic factor) conditioned on the conjecture that $(\min, +)$-convolution has no truly subquadratic-time algorithm, since this conjecture implies that Knapsack has no $O((n + 1/\varepsilon)^{2-\delta})$-time FPTAS for any constant $\delta > 0$.
\end{abstract}

\section{Introduction}
We consider the classic Knapsack problem. Given a knapsack of capacity $t$ and a set $I$ of $n$ items, where each item $i \in I$ has weight $w_i$ and profit $p_i$, our goal is to find a subset $I' \subseteq I$ that maximizes $\sum_{i \in I'}p_i$ subject to $\sum_{i \in I'} w_i \leq t$.

Knapsack is a fundamental problem in combinatorial optimization and belongs to Karp's 21 NP-complete problems~\cite{Kar72}. Consequently, extensive effort has been devoted to developing approximation algorithms for Knapsack. A fully polynomial-time approximation scheme (FPTAS) is an algorithm that takes a precision parameter $\eps$ and produces a solution that is within a factor $1+\eps$ from the optimum in time $\mathrm{poly}(n,1/\eps)$.  Since the first FPTAS for Knapsack, there has been a long line of research on improving its running time, as summarized in Table~\ref{table:Knapsack}.

\begin{table}[!ht]
    \centering
    \caption{Polynomial-time approximation schemes for Knapsack. Symbol (\dag) means that it is a randomized approximation scheme.}
    \begin{tabular}{cc}
        \hline
        Knapsack & Reference \\
        \hline\specialrule{0em}{0pt}{2pt}
        $n^{O(1/\eps)}$ & Sahni~\cite{Sah75} \\
        ${O}(n\log n+(\frac{1}{\eps})^4\log \frac{1}{\eps})$ & Ibarra and Kim~\cite{IK75} \\
        ${O}(n\log n+(\frac{1}{\eps})^4)$ & Lawler~\cite{Law79} \\
        $O(n\log \frac{1}{\eps}+(\frac{1}{\eps})^3\log^2\frac{1}{\eps})$ & Kellerer and Pferschy~\cite{KP04} \\ 
        $O(n\log \frac{1}{\eps}+(\frac{1}{\eps})^{5/2}\log^3\frac{1}{\eps})$ \dag
        & Rhee~\cite{Rhe15} \\
        $O(n\log \frac{1}{\eps}+(\frac{1}{\eps})^{12/5}/2^{\Omega(\sqrt{\log (1/\eps)})})$ & Chan~\cite{Chan18} \\
        $O(n\log \frac{1}{\eps}+(\frac{1}{\eps})^{9/4}/2^{\Omega(\sqrt{\log (1/\eps)})})$ & Jin~\cite{Jin19} \\
        $\widetilde{O}(n+(\frac{1}{\eps})^{11/5}/2^{\Omega(\sqrt{\log (1/\eps)})})$ \dag& Deng, Jin and Mao~\cite{DJM23} \\
        $\widetilde{O}(n + (\frac{1}{\eps})^2)$ & This Paper \\
        \specialrule{0em}{0pt}{2pt}\hline
    \end{tabular}
    \label{table:Knapsack}
\end{table}

Prior to our work, the best-known FPTAS for Knapsack runs in $\widetilde{O}(n+(\frac{1}{\eps})^{11/5})$ time\footnote{In the paper we use an $\widetilde{O}(\cdot)$ notation to hide polylogarithmic factors.}~\cite{DJM23}. On the negative side, Knapsack has no $O((n + 1/\eps)^{2-\delta})$-time FPTAS for any constant $\delta > 0$, assuming that $(\min, +)$-convolution has no truly subquadratic-time algorithm~\cite{CMWW19, KPS17}. The following open problem has been repeatedly mentioned in a series of papers~\cite{Chan18,Jin19,Jin24}.

 \begin{quote}
     Does Knapsack have a FPTAS that runs in $O(n + (\frac{1}{\eps})^2)$ time?
 \end{quote}

Very recently, important progress has been made on pseudo-polynomial-time exact algorithms for Knapsack. Assuming that all inputs are integers, Jin~\cite{Jin24} and Bringmann~\cite{Bri24} independently show that Knapsack can be solved in $\widetilde{O}(n + w_{\max}^2)$ time, where $w_{\max}$ refers to the largest item weight. Both works are built upon a recent paper by Chen, Lian, Mao, and Zhang~\cite{CLMZ24aSODA}. The main techniques of these papers are based on proximity results for Knapsack. Therefore, it is a natural idea to apply proximity results to designing approximation schemes. There is, however, a critical challenge in doing so. Roughly speaking, the proximity technique works for only a given capacity $t$. Meanwhile, the current best framework for approximating Knapsack requires solving Knapsack for all capacities up to $t$. As a consequence, a direct combination of the existing techniques does not yield a satisfactory running time.

We resolve this technical challenge and manage to apply proximity results to approximating Knapsack. Our main result is the following.

\begin{theorem}\label{thm:main}
    There is an FPTAS for Knapsack that runs in $\widetilde{O}(n + \frac{1}{\eps^2})$ time.
\end{theorem}

We remark that very recently, independent of our work, Mao~\cite{Mao24} also obtained an $\widetilde{O}(n + \frac{1}{\eps^2})$-time FPTAS for Knapsack using a different technique.

\subsection{Technique Overview}
We first briefly describe the proximity technique and the framework we will use to approximate Knapsack.

\paragraph{Proximity Results for Knapsack} A proximity result states that there is an optimal solution close to the greedy solution that selects the items in decreasing order of efficiency (profit-to-weight ratio). A more intuitive interpretation of proximity results is that an optimal solution will select many high-efficiency items but only a few low-efficiency items. The first proximity result was implied by Eisenbrand and Weismantel's work for a general integer program~\cite{EW19}, and its proof was later simplified in~\cite{PRW21} in the context of Knapsack. 

Chen, Lian, Mao, Zhang~\cite{CLMZ24aSODA} developed a stronger proximity result via additive combinatorics tools. We briefly discuss their result. Assume that all item profits are integers and that the items are labelled in decreasing order of efficiency. That is, $\frac{p_1}{w_1} \geq \frac{p_2}{w_2} \geq \cdots \geq \frac{p_n}{w_n}$. Let $b$ be the smallest index such that $\sum_{i=1}^b w_i > t$. In Knapsack terminology, $b$ is called the break item for $t$. Let $p_{\max}$ be the maximum item profit. Let $\ell$ be largest index such that the items in $\{\ell+1,\ldots,b - 1\}$ have $\Omega(p_{\max}^{1/2})$ distinct profits, and let $r$ be the smallest index such that the items in $\{b,\ldots,r-1\}$ have $\Omega(p_{\max}^{1/2})$ distinct profits. Let $I_H = \{1, \ldots, \ell\}$ be the set of high-efficiency items, and let $I_L = \{r, \ldots, n\}$ be the set of low-efficiency items. It can be shown that in an optimal solution, the low-efficiency items in $I_L$ contribute a total profit of at most $O(p^{3/2}_{\max})$, while the high-efficiency items in $I_H$ contribute a total profit of at least $\sum_{i \in I_H}p_i - O(p^{3/2}_{\max})$.  Bringmann~\cite{Bri24} and Jin~\cite{Jin24} further improved the proximity result by considering the multiplicity of profits and carefully dividing items into polylogarithmic subsets. We will use a proximity result that is slightly more general than the one by Chen et al.~\cite{CLMZ24aSODA}, but much simpler than those by Bringmann~\cite{Bri24} and Jin~\cite{Jin24}.

\paragraph{A Functional Approach for Approximating Knapsack} We leverage Chan's framework~\cite{Chan18} for approximating Knapsack. Let $(I,t)$ be a Knapsack instance. This framework approximately solves the problem for all capacities up to $t$.  In other words, it considers a more general problem of approximating a function $f_I$, where $f_I(x)$ is the maximum total profit that can be achieved given the capacity is $x$. The advantage of this framework is that it allows us to divide the items into groups, tackle each group independently, and then merge the resulting function using $(\max, +)$-convolution. Using this framework together with the standard scaling technique, we can reduce the original problem to instances where $p_i$ are integers in $(\frac{1}{2\eps}, \frac{1}{\eps}]$ and $f_I(t) \leq \frac{1}{\alpha\eps}$ for some $\alpha \in [1, \frac{1}{\eps}]$.

We combine the above two techniques to obtain a more efficient approximation algorithm. Although natural, it is not trivial. As we mentioned before, the proximity result works for only a single capacity $t$: the definition of high- and low-efficiency items (that is, $I_H$ and $I_L$) depends on the break item for $t$. Meanwhile, the function approach requires us to solve the reduced problem for all capacities up to $t$. Clearly, we cannot bear the cost of applying the proximity result for each capacity $x \in [0,t]$. We address this technical challenge as follows.

We first identify a set $I_{\mathrm{tail}}$ of items that can always be considered as low-efficiency items for all capacities up to $t$. These items can be handled efficiently since their contribution to the optimal solution is small no matter what the capacity is.

It remains to deal with items in $I_{\mathrm{head}} = I \setminus I_{\mathrm{tail}}$. The items in $I_{\mathrm{head}}$ can be high-efficiency, median-efficiency, or low-efficiency depending on what the capacity is. We carefully divide the interval $[0, t]$ into $O(\frac{1}{\alpha\eps^{1/2}})$ sub-intervals so that for each sub-interval $j$, we can identify two subsets $I^j_H$ and $I^j_L$ of items that can be considered as high-efficiency and low-efficiency, respectively, for all capacities within this interval. Let $I^j_M$ be the rest of the items. Now we have $O(\frac{1}{\alpha\eps^{1/2}})$ subsets of $I_{\mathrm{head}}$. All these subsets can dealt with efficiently due to the following reasons.
\begin{itemize}
    \item When a set has a small contribution to the optimal solutions, we can approximate this set with less precision and therefore with less time cost. When a set has a large contribution to the optimal solutions, we consider the ``complementary`` problem of determining which items should not be selected. Since the items not selected by the optimal solutions have a small total profit, we can approximate with less precision and therefore with less time cost.

    \item There are only $O(\frac{1}{\alpha\eps})$ items in $I_{\mathrm{head}}$, and therefore, the standard dynamic programming for Knapsack can be efficient (given that the optimal value is also small). Note that dynamic programming solves the problem for all prefixes of the items set and all capacities up to $t$.
    In particular, our division of the interval $[0,t]$ guarantees that $I^1_L \supseteq I^2_L \supseteq\cdots \supseteq I^\ell_L$.  Therefore, a single run of dynamic programming for $(I^1_L, t)$ can solve the problem for all $I^j_L$ and all capacities up to $t$. All the subsets $I^j_H$ can be handled similarly.

    \item Our method can also guarantee that for each subset $I^j_M$,  the items in this subset have at most $\widetilde{O}(\frac{1}{\eps^{1/2}})$ distinct profits. There are known algorithms for such item sets whose items have only a small number of distinct profits.   
\end{itemize}

\subsection{Further Related Work} \paragraph{Running Time of Other Forms} In addition to the running time of the form $n+f(1/\eps)$, there is also a line of research focusing on FPTASes whose running time is of the form $n^{O(1)}\cdot f(1/\eps)$. Lawler~\cite{Law79} presented an FPTAS of running time $O(n^2/\eps)$. Kellerer and Pferschy~\cite{KP99} presented an alternative FPTAS of running time $\widetilde{O}({n}/{\eps^2})$. Chan~\cite{Chan18} gave two FPTASes of running time $\widetilde{O}({n^{3/2}}/{\eps})$ and $O((n/\eps^{4/3}+1/\eps^2)/2^{\Omega(\sqrt{\log (1/\eps)})})$, respectively. Jin~\cite{Jin19} gave an FPTAS of running time $O((n^{3/4}/\eps^{3/2}+1/\eps^2)/2^{\Omega(\sqrt{\log (1/\eps)}}+n\log (1/\eps))$. It remains open whether there exists an FPTAS of running time $\widetilde{O}(n/\eps)$.

\paragraph{Unbounded Knapsack} A problem closely related to 0-1 Knapsack is Unbounded Knapsack, where each item can be chosen arbitrarily many times.  There exists an FPTAS of running time $\widetilde{O}(n+\frac{1}{\eps^2})$ for the unbounded knapsack by Jansen and Kraft~\cite{JK18}, as well as by Bringmann and Cassis~\cite{BC22}. Like Knapsack, Unbounded Knapsack has no $O((n + 1/\eps)^{2-\delta})$-time FPTAS for any constant $\delta > 0$, assuming that $(\min, +)$-convolution has no truly subquadratic-time algorithm~\cite{CMWW19, KPS17}.

\paragraph{Additive Combinatorics} A fundamental result in additive combinatorics is that if a set $X\subset [n]$ contains $\Omega(\sqrt{n})$ distinct integers, then the set of all subset-sums of $X$ contains an arithmetic progression of length $\Omega(n)$ (see, e.g.~\cite{Fre93,Sar89,SV06,CFP21}). This has led to several significant progress in pseudo-polynomial time algorithms for the Subset-Sum problem~\cite{CFG89,GM91, Cha99,BW21}. Very recently, researchers further extended this method to obtain faster approximation and exact algorithms for Knapsack~\cite{PRW21,DJM23,CLMZ24aSODA,Bri24,Jin24}. 
These results rely on additive combinatorics methods to obtain a good proximity bound for Knapsack.  

\subsection{Paper Outline}
Section~\ref{sec:notation} defines some necessary terminology. Section~\ref{sec:functional} introduces the framework we will use for approximating Knapsack. In Section~\ref{sec:reduce}, we show how to reduce the original problem to instances with nice properties. Section~\ref{sec:proximity} gives a proximity result for Knapsack. In Section~\ref{sec:alg}, we present our algorithm for the reduced problem. We conclude this paper in Section~\ref{sec:conclusion}.

\section{Notation}\label{sec:notation}
Let $I$ be a set of items. We use $p(I)$ and $w(I)$ to denote the total profit and the total weight of the items in $I$, respectively. That is, $p(I) = \sum_{i \in I} p_i$ and $w(I) = \sum_{i \in I}w_i$. The efficiency of an item $i$ is defined to be $e_i = \frac{p_i}{w_i}$.

Let $f: \mathbb{R} \to \mathbb{R}$ be a function. Let $B$ be a real number. We define $\mathrm{leq}(f, B)$ to be the function whose value on $x$ is $f(x)$ if $f(x) \leq B$ and is $-\infty$ if $f(x) > B$. Similarly, we define $\mathrm{geq}(f, B)$ to be the function whose value on $x$ is $f(x)$ if $f(x) \geq B$ and is $-\infty$ if $f(x) < B$.  Basically, $\mathrm{leq}(f, B)$ represents the part of $f$ whose function values are at most $B$, and $\mathrm{geq}(f, B)$ represents the part of $f$ whose function values are at least $B$.

\section{Approximating for All Capacities}\label{sec:functional}
We adopt the ``functional'' approach used in~\cite{Chan18,Jin19}. Let $(I, t)$ be a Knapsack instance. Instead of (approximately) solving the problem for a single capacity $t$, we consider all capacities up to $t$. More precisely, we consider the problem of approximating the following function
\[
    f_{I}(x) = \max\left\{\textrm{$p(I')$ : $I' \subseteq I$ and $w(I') \leq x$}\right\} 
\]
for all $x \in [0, t]$. (The function values for $x > t$ do not matter. For simplicity, we assume that all functions in this article are defined on $[0, t]$.) It is easy to see that $f_I$ is a monotone (non-decreasing) step function. We say that a function $f'$ approximates $f$ with factor $1 + \eps$ if $1 \leq \frac{f(x)}{f'(x)} \leq 1 + \eps$ for all $x \in [0, t]$. We say that $f'$ approximates $f$ with additive error $\delta$ if $0\leq f(x) - f'(x) \leq \delta$ for all $x \in [0,t]$. We remark that additive error is what we will eventually bound, but sometimes we still use approximation with a factor because it is more standard and makes the analysis easier. It helps to note that when $f$ is a monotone function, an approximation with factor $1 + \eps$ implies an approximation with additive error $\eps f(t)$.

The advantage of the functional approach is that it allows us to divide the items into groups, tackle each group separately, and merge the resulting functions via $(\max, +)$-convolution.  Let $I_1 \cup I_2$ be a partition of $I$. It is easy to observe that 
\[
    f_I(x) = \max_{x' \in [0,x]}(f_{I_1}(x') + f_{I_2}(x - x'))
\]
for $x \in [0, t]$. In other words, $f_I = f_{I_1} \oplus f_{I_2}$, where $\oplus$ denotes the $(\max, +)$-convolution. We make the observation that the $(\max,+)$-convolution of approximations of two functions is an approximation of their $(\max,+)$-convolution. That is,
\begin{itemize}
    \item If $f'$ and $g'$ approximate $f$ and $g$ with factor $1 + \eps$, respectively, then $f' \oplus g'$ approximate $f \oplus g$ with factor $1 + \eps$.

    \item If $f'$ approximates $f$ with additive error $\delta_1$, and $g'$ approximates $g$ with additive error $\delta_2$, then $f' \oplus g'$ approximate $f \oplus g$ with additive error $\delta_1  + \delta_2$.
\end{itemize}

The $(\max, +)$-convolution of two monotone step functions can be computed efficiently when two functions have their function values being integers, and one of them is $p$-uniform and pseudo-concave. A step function is $p$-uniform if its function values are of the form $0, p, 2p, \ldots, \ell p$. A $p$-uniform step function is pseudo-concave if  $x_{i+1} - x_{i} \geq x_{i} - x_{i-1}$ for every integer $i \in [1, \ell-1]$, where $x_i$ is the $x$-point at which the function value changes to $ip$.  In particular, when all the items in $I$ have the same profit $p$, the function $f_I$ is $p$-uniform and pseudo-concave.

\begin{lemma}\label{lem:exact-two-conv}
    Let $f$ and $g$ be two monotone step functions whose function values are integers in $[0,B]$. If $g$ is pseudo-concave and $p$-uniform for some $p > 0$, then we can compute $f\oplus g$ in $O(B)$ time.
\end{lemma}
\begin{proof}
    The proof is similar to that of~\cite[Fact 1]{Chan18}. For $y \in \{0, 1, \ldots, B\}$, define $f^{-1}(y)$ be the largest $x$ with $f(x) \leq y$. If no such $x$ exist, define $f^{-1}(y) = 0$. Define $g^{-1}(y)$ similarly. $f^{-1}$ and $g^{-1}$ are actually two sequences of length $B$.  We can compute their $(\min, +)$-convolution 
    \[
        (f \oplus g)^{-1}(y)= \min_{y'\in \{0,1, \ldots, y\}}\{f^{-1}(y') + g^{-1}(y - y')\}
    \]
    in $O(B)$ time, using an algorithm similar to that in~\cite[Lemma 10]{AT19}. From $(f \oplus g)^{-1}$, we can obtain $f \oplus g$ in $O(B)$ time.
\end{proof}

We say that a step function is of complexity $\ell$ if it has $\ell$ steps. There are also algorithms for approximating $(\max, +)$-convolution.

\begin{lemma}[{\cite[Lemma~2]{Chan18}}]\label{lem:approx-m-conv}
    Let $f_1, \ldots, f_m$ be monotone step functions with total complexity $\ell$ and ranges contained in $\{-\infty, 0\} \cup [A, B]$. Then we can compute a monotone step function that approximates $f_1 \oplus \cdots \oplus f_m$ with factor $1 + O(\eps)$ and complexity $\widetilde{O}(\frac{1}{\eps})$ in 
    \begin{enumerate}[label={\normalfont(\roman*)}]
        \item  $O(\ell)  + \widetilde{O}(\frac{1}{\eps^2}m)$ time in general;
        \item  $O(\ell) + \widetilde{O}(\frac{1}{\eps}m^2)$ time if, for $i \geq 2$, $f_i$ is $p_i$-uniform and pseudo-concave for some $p_i$.
    \end{enumerate}
\end{lemma}

We remark that the running time in Lemma~\ref{lem:approx-m-conv} hides polylogarithmic factors in $\frac{B}{A}$. Throughout this article, we will always have $A \geq 1$ and $B \leq \frac{2}{\eps^2}$, so $\frac{B}{A}$ is bounded by $\frac{2}{\eps^2}$.

\section{Reducing the Problem}\label{sec:reduce}
We shall show that the functional approach in the previous section, together with some standard scaling techniques, makes it sufficient to consider the following reduced problem $\mathrm{RP}(\eps, \alpha)$.

\begin{quote}
    Let $\eps \in (0, 1]$ and $\alpha \in [1, \frac{1}{\eps}]$ be two numbers. Given a capacity $t$ and a set ${I}$ of $n$ items where every $p_i$ is an integer in $[\frac{1}{2\eps}, \frac{1}{\eps}]$ and $f_{I}(t) \leq \frac{1}{\alpha\eps^2}$, compute a function of complexity $\widetilde{O}(\frac{1}{\eps})$ that approximates $f_{I}$ with additive error $\widetilde{O}(\frac{1}{\alpha\eps})$.
\end{quote}

Let $(I, t)$ be a Knapsack instance. Our ultimate goal is to approximate a single value $f_I(t)$ with factor $1 + \eps$. Obviously, this can be reduced to approximating the function $f_I$ with additive error $\frac{\eps}{2}f_I(t)$.  We can relax the permissible additive error to $\widetilde{O}(\eps f_I(t))$, because an additive error of $\widetilde{O}(\eps f_I(t))$ can always be reduced to $\frac{\eps}{2}f_I(t)$ by adjusting $\eps$ by a polylogarithmic factor, which increases the running time by only a polylogarithmic factor. Now our goal becomes approximating $f_I$ with additive error $\widetilde{O}(\eps f_I(t))$.

\paragraph{Scaling the Profits} We scale the item profits to make $f_{I}(t) \in [\frac{1}{\eps^2}, \frac{2}{\eps^2}]$,  as follows. It is well-known that Knapsack has a simple linear-time 2-approximation algorithm~\cite[Section 2.5]{KPD04}. Using this algorithm, we can obtain a number $y$ such that $y \leq f_{I}(t) \leq 2y$ in $O(n)$ time. By scaling all profits $p_i$ by a factor of $y\eps^2$, we have that $f_{I}(t) \in [\frac{1}{\eps^2}, \frac{2}{\eps^2}]$, and the permissible additive error becomes $\widetilde{O}(\eps f_I(t)) = \widetilde{O}(\frac{1}{\eps})$.

\paragraph{Dealing with Tiny Items} We further simplify the instance to make every $p_i \in (\frac{1}{\eps}, \frac{2}{\eps^2}]$.  All the items with $p_i > \frac{2}{\eps^2}$ can be safely discarded as $f_{I}(t) \leq \frac{2}{\eps^2}$. We say that an item $i$ is tiny if $p_i \leq \frac{1}{\eps}$. We partition all the tiny items into groups $S_1, \ldots, S_\ell$ such that the following holds.
\begin{enumerate}[label = {(\roman*)}]
    \item For every integer $j \in [1, \ell-1]$, the efficiency of every item in $S_j$ is no smaller than that of any item in $S_{j+1}$. That is, $e_i \geq e_{i'}$ for any $i \in S_j$ and any $i' \in S_{j+1}$.

    \item $\frac{1}{\eps} < p(S_j) \leq \frac{2}{\eps}$ for every integer $j \in [1, \ell - 1]$ and $p(S_\ell) \leq \frac{2}{\eps}$.
\end{enumerate} 
This can be easily done by scanning all the tiny items in decreasing order of their efficiencies. If $p(S_\ell) \leq \frac{1}{\eps}$, we simply discard all the items in $S_\ell$, and this increases the additive error by at most $\frac{1}{\eps}$. For each remaining $S_j$, we replace it with a single meta item with profit $p(S_j)$ and weight $w(S_j)$. Note that $p(S_j) \in (\frac{1}{\eps}, \frac{2}{\eps}] \subseteq (\frac{1}{\eps}, \frac{2}{\eps^2}]$. We claim that replacing the tiny items with meta items incurs at most an additive error $\frac{2}{\eps}$. To see this, it suffices to show that for any set $S$ of tiny items, there exists a set $S'$ of meta items such that $p(S') > p(S) - \frac{2}{\eps}$ and $w(S') \leq w(S)$. Indeed, $S'$ can be obtained as follows: select the meta items greedily in decreasing order of efficiency and stop immediately when adding the next item to $S'$ makes its total weight greater than $w(S)$. Clearly, $w(S') \leq w(S)$. To bound $p(S')$, consider the set $S''$ that is obtained by greedily adding one more meta item to $S'$. Clearly, $w(S'') > w(S)$. Moreover, by the construction of meta items, $S''$ can also be viewed as a set of tiny items that are selected greedily in decreasing order of efficiency. Therefore, $p(S'') > p(S)$. Note that $S'$ and $S''$ differ by exact one meta item, so $p(S') \geq p(S'') - \frac{2}{\eps} > p(S) - \frac{2}{\eps}$.

\paragraph{Grouping the Items}   Since every $p_i \in (\frac{1}{\eps}, \frac{2}{\eps^2}]$, we can divide items into $1 + \log \frac{1}{\eps}$ groups such that items in the $j$-th group have $p_i \in (\frac{2^{j-1}}{\eps}, \frac{2^j}{\eps}]$. We can approximate each group separately, and merge the resulting functions via Lemma~\ref{lem:approx-m-conv}(i).  Suppose that for each group $I'$, we can compute a function of complexity $\widetilde{O}(\frac{1}{\eps})$ that approximates $f_{I'}$ with additive error $\widetilde{O}(\frac{1}{\eps})$. Since there are only $O(\log \frac{1}{\eps})$ groups, it takes only $\widetilde{O}(\frac{1}{\eps^2})$ time to merge these functions using Lemma~\ref{lem:approx-m-conv}(i), and the merged function approximates $f_I$ with additive error $\widetilde{O}(\frac{1}{\eps})$.

Now we can focus on a single group ${I}'$ of items. We have that $p_i \in (\frac{\alpha}{\eps}, \frac{2\alpha}{\eps}]$ for some $\alpha \in [1, \frac{1}{\eps}]$ and that $f_{{I}'}(t) \leq f_{I}(t) \leq \frac{2}{\eps^2}$. By dividing all profits $p_i$ by a factor of $2\alpha$, we have that $p_i \in (\frac{1}{2\eps}, \frac{1}{\eps}]$, and that $f_{{I}'}(t) \leq \frac{1}{\alpha\eps^2}$. The permissible additive error becomes $\widetilde{O}(\frac{1}{\alpha\eps})$. We assume that $\frac{1}{2\eps}$ is an integer, and this can be done by adjusting $\eps$ by a constant factor. We further assume that all profits $p_i$ are integers, and this changes the optimal value by at most a factor of $1 + 2\eps$, and therefore, incurs an additive error of at most $O(\eps f_I(t)) = O(\frac{1}{\eps})$.

We conclude this section by the following lemma.

\begin{lemma}\label{lem:reduction}
    An $\widetilde{O}(n + \frac{1}{\eps^2})$-time algorithm for the reduced problem $RP(\eps, a)$ implies an $\widetilde{O}(n + \frac{1}{\eps^2})$-time FPTAS for Knapsack.
\end{lemma}

\section{A Proximity Result}\label{sec:proximity}
Proximity Results for Knapsack capture the intuition that an optimal solution should select many high-efficiency items but only a few low-efficiency items. The proximity result we shall derive in this section is slightly more general than the one by Chen et al.~\cite{CLMZ24aSODA}, but much simpler than those by Bringmann~\cite{Bri24} and Jin~\cite{Jin24}. We provide a proof for completeness, although in essence, it is the same as that in~\cite{CLMZ24aSODA}.

Let $I$ be a set of items labelled in decreasing order of efficiency. That is, $e_1 \geq e_2 \geq \cdots \geq e_n$. For any capacity $t$, the break item for $t$ is the first item $b$ with $\sum_{i=1}^b w_i > t$. In case that $\sum_{i=1}^n w_i \leq t$, we set $b = n + 1$. Let $I' \subseteq I$ be a subset of items. Let $i'_1$ and $i'_2$ be the items in $I'$ with the smallest and largest index, respectively. We say that 
\begin{itemize}
    \item $I'$ is superior to $b$ by $\tau$ distinct profits if $i'_2 < b$ and the items in $\{i : i'_2 < i < b\}$ have at least $\tau$ distinct profits, and that

    \item $I'$ is inferior to $b$ by $\tau$ distinct profits if $i'_1 > b$ and the items in $\{i : b \leq i < i'_1\}$ have at least $\tau$ distinct profits.
\end{itemize}
See Figure~\ref{fig:sep} for an example.  

\begin{figure}[ht]
        \centering
        \begin{tikzpicture}
        \draw [thick] (0,0) -- (8,0);
        \node [align=center, left] at (-0.2,0) {$I$};

        \draw [thick] (0, 0.1) -- (0, -0.1);
        \node [align =center, below] at (0, -0.2) {$1$};

        \draw [thick] (4, 0.1) -- (4, -0.1);
        \node [align =center, below] at (4, -0.2) {$b$};

        \draw [thick] (8, 0.1) -- (8, -0.1);
        \node [align =center, below] at (8, -0.2) {$n$};

        \draw [thick] (6, 0.1) -- (6, -0.1);
        \node [align = center, rotate = -90, yscale = 5.5] at (5, 0.4) {$\{$};
        \node [align = center, above] at (5, 0.4) {\scriptsize {$\tau$ distinct}\\{\scriptsize profits}};
        
        \draw [thick] (2, 0.1) -- (2, -0.1);
        \node [align = center, rotate = -90, yscale = 5.5] at (3, 0.4) {$\{$};
        \node [align = center, above] at (3, 0.4) {\scriptsize {$\tau$ distinct}\\{\scriptsize profits}};

        \draw [thick] (1, 0.1) -- (1, -0.1);
        \node [below] at (1.5, -0.1) {$I_H$};

        \draw [thick] (7, 0.1) -- (7, -0.1);
        \node [below] at (6.5, -0.1) {$I_L$};

        \draw [thick, ->] (1, - 1) -- (7, - 1);
        \node [align = center, below] at (4,-1.2) {larger index implies lower efficiency};

        \end{tikzpicture}
        \caption{$I_H$ is superior to $b$ and $I_L$ is inferior to $b$ by $\tau$ distnct profits}
        \label{fig:sep}
\end{figure}

Basically, items superior to $b$ are high-efficiency items, and items inferior to $b$ are low-efficiency items. The following lemma gives the proximity result we need.
\begin{lemma}\label{lem:proximity}
    There exists a constant $c_p$ such that the following is true. Let $(I, t)$ be a Knapsack instance where items are labelled in decreasing order. Assume that the item profits are integers contained in $[1, p_{\max}]$. Let $b$ be the break item for $t$. There is an optimal solution $I^*$ such that for any $k \geq 1$ and any $I' \subseteq I$,
    \begin{enumerate}[label={\normalfont(\roman*)}]
        \item if $I'$ is superior to $b$ by $c_pkp^{1/2}_{\max}\log p_{\max}$ distinct profits, then
        \[
            p(I' \cap I^*) \geq p(I') - \frac{p_{\max}^{3/2}}{k};
        \]

        \item If $I'$ is inferior to $b$ by $c_pkp^{1/2}_{\max}\log p_{\max}$ distinct profits, then
        \[
            p(I' \cap I^*) \leq \frac{p_{\max}^{3/2}}{k}.
        \]
    \end{enumerate}
\end{lemma} 

The proof of Lemma~\ref{lem:proximity} relies on the following additive combinatoric result, which states that if we have two (multi-) sets of integers, one having lots of distinct values and the other having a large sum, then some non-empty subsets of these two (multi-) sets must have the same sum.

\begin{lemma}\label{lem:add-comb}
    There exists a constant $c$ such that the following holds. Let $X$ be a set of integers from $[1, p_{\max}]$. Let $Y$ be a multi-set of integers from $[1, p_{\max}]$. If for some $k \geq 1$, $|X| \geq ckp_{\max}^{1/2}\log p_{\max}$ and $\sum_{y \in Y} y \geq p_{\max}^{3/2}/k$, then there exists non-empty subset $X' \subseteq X$ and $Y' \subseteq Y$ such that $\sum_{x \in X'}x = \sum_{y \in Y'} y$. 
\end{lemma}

Assuming Lemma~\ref{lem:add-comb} holds, we prove Lemma~\ref{lem:proximity}.
\begin{proof}[Proof of Lemma~\ref{lem:proximity}]
    Let $c_p = 2c$ where $c$ is the constant in Lemma~\ref{lem:add-comb}. Each optimal solution can be viewed as a vector in $\{0,1\}^n$. Let $I^*$ be the optimal solution with the highest lexicographical order. That is, $I^*$ prefers items with smaller indices. We show that $I^*$ satisfies the state property.

    We first prove property (i). Consider arbitrary $I' \subseteq I$ that is superior to $b$ by $c_pkp^{1/2}_{\max}\log p_{\max}$ distinct profits. Let $i'$ be the item in $I'$ with the largest index. Define $I_H = \{i: 1 \leq i  \leq  i'\}$, $I_M = \{i: i' < i < b\}$ and $I_L = \{i: b \leq i \leq n\}$. Note that $I' \subseteq I_H$ and that the items in $I_M$ have $c_pkp^{1/2}_{\max}\log p_{\max}$ distinct profits. We shall show that $p(I_H \setminus I^*) \leq p^{3/2}_{\max}/k$, which will imply that
    \[
        p(I' \cap I^*) \geq p(I') - p(I' \setminus I^*) \geq  p(I') - p(I_H \setminus I^*) \geq p(I') - \frac{p_{\max}^{3/2}}{k}.
    \] 
    Suppose, for the sake of contradiction, that $p(I_H \setminus I^*) > p^{3/2}_{\max}/k$. Consider $I_M$. Recall that $c_p = 2c$. Either 
    \begin{enumerate}[label= {(\alph*)}]
        \item the items in $I_M \cap I^*$ have at least $ckp^{1/2}_{\max}\log p_{\max}$ distinct profits, or

        \item the items in $I_M \setminus I^*$ have at least $ckp^{1/2}_{\max}\log p_{\max}$ distinct profits.
    \end{enumerate}
    
    In case (a), the collection of profits of the items in $I_H \setminus I^*$ can be viewed as a multi-set $Y$ with $\Sigma_{y \in Y} y >  p^{3/2}_{\max}/k$, and the collection of distinct profits of the items in $I_M \cap I^*$ can be viewed as a set $X$ with $|X| \geq ckp^{1/2}_{\max}\log p_{\max}$. By Lemma~\ref{lem:add-comb}, there are non-empty subsets $X' \subseteq X$ and $Y' \subseteq Y$ such that $\sum_{x \in X'}x = \sum_{y \in Y'} y$. This implies that there are non-empty subsets $I^+_H \subseteq I_H \setminus I^*$ and $I_M^- \subseteq I_M \cap I^*$ with $p(I^+_H) = p(I_M^-)$. Also, $w(I^+_H) \leq w(I^-_M)$ because the efficiency of any items in $I_H$ is no smaller than that of an item in $I_M$. Therefore, by replacing $I^-_M$ with $I^+_H$, we can obtain another optimal solution $I^{**} = (I^* \setminus I^-_M) \cup (I^+_H)$. Note that $I^{**}$ has a higher lexicographical order than $I^*$. But $I^*$ is the optimal solution with the highest lexicographical order, which makes a contradiction.

    Consider case (b). By the definition of break item, we have 
    \(
        p(I^*) \geq p(I_H) + p(I_M).
    \)
    Therefore, $p(I_H \setminus I^*) > p^{3/2}_{\max}/k$ implies that $p(I_L \cap I^*) > p^{3/2}_{\max}/k$. By an argument similar to that in case (a), there are non-empty subsets $I_M^+ \subseteq I_M \setminus I^*$ and $I_L^- \subseteq I_L \cap I^*$ such that $p(I_M^+) = p(I_L^-)$. By replacing $I_L^-$ with $I_M^+$, we can obtain an optimal solution $I^{**} = (I^* \setminus I^-_L) \cup (I^+_M)$ that has a high lexicographical order than $I^*$. Contradiction.

    Property (ii) can be proved similarly. Consider arbitrary $I' \subseteq I$ that is inferior to $b$ by $c_pkp^{1/2}_{\max}\log p_{\max}$ distinct profits. Let $i'$ be the item in $I'$ with the smallest index. Define $I_H = \{i: 1 \leq i < b\}$, $I_M = \{i : b \leq i < i'\}$ and $I_L = \{i : i'\leq i\leq n\}$. Note that $I' \subseteq I_L$ and that the items in $I_M$ have $c_pkp^{1/2}_{\max}\log p_{\max}$ distinct profits. To prove property (ii), it suffices to show that $p(I_L \cap I^*) \leq p^{3/2}_{\max}$. 
    Suppose, for the sake of contradiction, that $p(I_L \cap I^*) > p^{3/2}_{\max}/k$. Consider $I_M$. Either
    \begin{enumerate}[resume*]
        \item the items in $I_M \setminus I^*$ have at least $ckp^{1/2}_{\max}\log p_{\max}$ distinct profits, or

        \item the items in $I_M \cap I^*$ have at least $ckp^{1/2}_{\max}\log p_{\max}$ distinct profits.
    \end{enumerate}
    In case (c), by an argument similar to that in case (a), there are non-empty subsets $I^+_M \subseteq I_M \setminus I^*$ and $I^-_L \subseteq I_L \cap I^*$ such that $p(I^+_M) = p(I^-_L)$. By replacing $I_L^-$ with $I_M^+$, we can obtain an optimal solution that has a higher lexicographical order than $I^*$. Contradiction.

    In case (d), $p(I_L \cap I^*) > p^{3/2}_{\max}/k$ implies that $p(I_H \setminus I^*) > p^{3/2}_{\max}/k$. Then there are non-empty subsets $I^+_H \subseteq I_H \setminus I^*$ and $I_M^- \subseteq I_M \cap I^*$ with $p(I^+_H) = p(I_M^-)$. By replacing $I_M^-$ with $I_H^+$, we can obtain an optimal solution that has a higher lexicographical order than $I^*$. Contradiction.
\end{proof}

\subsection{Tools from Additive Combinatorics}
Lemma~\ref{lem:add-comb} is a consequence of several results in~\cite{BW21}. We remark that the results in~\cite{BW21} apply to multi-sets, but what we present below is their set version. Let $X$ be a set of integers. We denote the sum of $X$ as $\Sigma(X)$, the maximum element in $X$ as $\max(X)$.

\begin{definition}[{\cite[Definition 3.1]{BW21}}]
\label{def:dense}
    We say that a set $X$ is $\delta$-dense if it satisfies {$|X|^2 \geq \delta \cdot \max(X)$}.
\end{definition}

\begin{definition}[{\cite[Definition 3.2]{BW21}}]
\label{def:divisor}
    We denote by $X(d)$ the set of all integers in $X$ that are divisible by $d$. Further, we write $\overline{X(d)}:= X \backslash X(d)$ to denote the set of all integers in $X$ not divisible by $d$. We say an integer $d > 1$ is an $\alpha$-almost divisor of $X$ if $|\overline{X(d)}|\leq\alpha\Sigma(X)/|X|^2$.
\end{definition}

\begin{theorem}[{\cite[Theorem 4.1]{BW21}}]
\label{thm:divisor}
    Let $\delta,\alpha$ be functions of $n$ with $\delta\ge1$ and $16\alpha\leq\delta$. Given a $\delta$-dense set $X$ of size $n$, there exists an integer $d\geq 1$ such that $X':= X(d)/d$ is $\delta$-dense and has no $\alpha$-almost divisor. Moreover, we have the following additional properties:
\begin{enumerate}[label={\normalfont (\roman*)}]
    \item $d \leq 4\Sigma_X/|X|^2$, 
    \item $|X'| \geq 0.75 |X|$,
    \item $\Sigma_{X'}\ge0.75\Sigma_X/d$.
\end{enumerate}
\end{theorem}

\begin{theorem}[{~\cite[Theorem 4.2]{BW21}}]\label{thm:hitrange}
    Let $X$ be a set of positive integers and set
    \begin{align*}
        c_\delta &:= 1699200\cdot \log(2|X|), \\
        c_\alpha &:= 42480,\\
        c_\lambda &:= 169920,\\
        \lambda_X &:= c_\lambda \max(X) \Sigma(X)/|X|^2.  
    \end{align*}
    If $X$ is $c_\delta$-dense and has no $c_\alpha$-almost divisor, then for every integer 
    \(
        s \in [ \lambda_X, \Sigma(X) -\lambda_X ],
    \) 
    there is a subset $X' \subseteq X$ with $\Sigma(X') = s$.
\end{theorem}

Now we are ready to prove Lemma~\ref{lem:add-comb}.
\begin{proof}[Proof of Lemma~\ref{lem:add-comb}]
    Assume that the constant $c$ is sufficiently large. Let $c_\delta$, $c_\alpha$, and $c_\lambda$ be defined as in Theorem~\ref{thm:hitrange}. Since $|X| \geq ckp^{1/2}_{\max}\log p_{\max}$, we have that 
    \[
        |X|^2 \geq c^2k^2p_{\max}\log p_{\max} \geq c^2p_{\max}\log |X| \geq c_\delta p_{\max}.
    \]
    That second inequality is due to that $k \geq 1$ and $|X| \leq p_{\max}$. By Definition~\ref{def:dense}, $X$ is $c_\delta$-dense. By Theorem~\ref{thm:divisor}, there exists an integer $d \geq 1$ such that $X':= X(d)/d$ is $c_{\delta}$-dense and has no $c_\alpha$-almost divisor. By Theorem~\ref{thm:hitrange}, every integer in $[\lambda_{X'}, \Sigma(X') - \lambda_{X'}]$ is a subset sum of $X'$. Therefore, every multiple of $d$ in the range $[d \lambda_{X'}, d\Sigma(X') - d\lambda_{X'}]$ is a subset sum of $X$.

    \begin{claim}\label{clm:add-comb}
        $d \lambda_{X'} \leq \frac{p^{3/2}_{\max}}{2k}$ and $d\Sigma(X') - 2d\lambda_{X'} \geq dp_{\max}$.
    \end{claim}
    \begin{proof}
        By Theorem~\ref{thm:hitrange} and the definition of $X'$, 
        \[
            \lambda_{X'} \leq \frac{c_\lambda\max(X')\Sigma(X')}{|X'|^2} \leq \frac{c_\lambda\max(X)\Sigma(X)}{d^2|X'|^2}.
        \]
        Theorem~\ref{thm:divisor}(ii) guarantees that $|X'| \geq 0.75|X|$. Therefore,
        \begin{equation}\label{eq:add-claim}
            \lambda_{X'}\leq \frac{c_\lambda\max(X)\Sigma(X)}{d^2(0.75|X|)^2}  \leq \frac{c}{2}\cdot \frac{p_{\max}\Sigma(X)}{d^2|X|^2}.
        \end{equation}
        The second inequality is due to that $c$ is sufficiently large.  Then we have
        \[
            d\lambda_{X'}  \leq  \frac{cp_{\max}\Sigma(X)}{2d|X|^2}
             \leq \frac{cp^2_{\max}}{2d|X|} \leq \frac{cp^2_{\max}}{2|X|} \leq  \frac{cp^2_{\max}}{2kcp^{1/2}_{\max}\log p_{\max}} \leq \frac{p^{3/2}_{\max}}{2k}.
        \]
        The second inequality is due to that $\Sigma(X) \leq p_{\max}|X|$. Now consider $d\Sigma(X') - 2d\lambda_{X'}$. Theorem~\ref{thm:divisor}(iii) implies that $d\Sigma(X') \geq 0.75\Sigma(X)$. In view of~\eqref{eq:add-claim}, we have
        \begin{align*}
            d\Sigma(X') - 2d\lambda_{X'} &\geq \frac{3\Sigma(X)}{4} - \frac{cp_{\max}\Sigma(X)}{d|X|^2}\\
                                        &\geq \frac{4\Sigma(X)}{|X|^2}\left(\frac{3|X|^2}{16} - \frac{cp_{\max}}{4d}\right)\\
                                        &\geq \frac{4\Sigma(X)}{|X|^2}\left(\frac{3c^2p_{\max}\log^2p_{\max}}{16} - \frac{cp_{\max}}{4d}\right)\\
                                        & \geq dp_{\max}.
        \end{align*} 
        The last inequality is due to Theorem~\ref{thm:divisor}(i) and that $c$ is sufficiently large.
    \end{proof}

    To complete the proof, it suffices to show that there is a subset $Y'$ of $Y$ whose sum is a multiple of $d$ and is within the interval $[d\lambda_{X'}, d\Sigma(X') - d\lambda_{X'}]$. We claim that as long as $Y$ contains at least $d$ integers, there must be a non-empty subset $Y' \subseteq Y$ whose sum is at most $dp_{\max}$ and is a multiple of $d$. Assume the claim is true. We can repeatedly extract such subsets $Y_1, Y_2, \ldots, Y_\ell$ from $Y$ until $Y$ has less than $d$ integers. Note that the total sum of these subsets 
\begin{align*}
  \Sigma(Y_1\cup \cdots \cup Y_{\ell}) &\geq \Sigma(Y) - d p_{\max} \\
    &\geq \frac{p_{\max}^{3/2}}{k} - \frac{4\Sigma(X)}{|X|^2}\cdot p_{\max}\\
    &\geq \frac{p_{\max}^{3/2}}{k} -  \frac{4p_{\max}^2}{|X|} \geq \frac{p_{\max}^{3/2}}{k} - \frac{4p^{3/2}_{\max}}{ck\log p_{\max}} \geq \frac{p_{\max}^{3/2}}{2k} \geq d\lambda_X'.
\end{align*}
The second inequality is due to Theorem~\ref{thm:divisor}(i), and the last is due to Claim~\ref{clm:add-comb}. This inequality implies that there is some index $j$ such that $\Sigma(Y_1 \cup \cdots \cup Y_j) \geq d\lambda_{X'}$. Let $j^*$ be the smallest such index. Since $\Sigma(Y_j) \leq dp_{\max}$ for each $Y_j$, it must be that 
\[
    \Sigma(Y_1 \cup \cdots \cup Y_{j^*}) \leq d\lambda_{X'} + dp_{\max} \leq d\Sigma(X') - d\lambda_{X'}.
\]
The second inequality is due to Claim~\ref{clm:add-comb}. Therefore, $Y_1 \cup \cdots \cup Y_{j^*}$ has its sum in the interval $[d\lambda_{X'}, d\Sigma(X') - d\lambda_{X'}]$. Its sum is also a multiple of $d$ since every $\Sigma(Y_j)$ is a multiple of $d$.

It remains to prove the claim that as long as $Y$ contains at least $d$ integers, there must be a non-empty subset $Y' \subseteq Y$ whose sum is at most $dp_{\max}$ and is a multiple of $d$. Take $d$ arbitrary integers from $Y$. For $i \in \{0,..., d\}$, Let $s_i$ be the sum of the first $i$ integers. By the pigeonhole principle, there must be $i < j$ such that $s_i \equiv s_j \pmod d$. This implies that there is a subset of $j - i$ integers whose sum is $s_j - s_i \equiv 0 \pmod d$. Note that $0 < j - i \leq d$. So this subset is non-empty and has its sum at most $dp_{\max}$.
\end{proof}

\section{The Algorithm}\label{sec:alg}
Throughout this section, we fix $(I, t)$ to be an instance of the reduced problem $\mathrm{RP}(\eps, \alpha)$, where the items are labelled in decreasing order of efficiency. Let $b$ be the break item for $t$. Since $f_{I}(t) \leq \frac{1}{\alpha\eps^2}$ and every $p_i \geq \frac{1}{2\eps}$, we have that
$b \leq \frac{2f_{I}(t)}{\eps} + 1 = \frac{2}{\alpha\eps} + 1$.  Let $\tau = \lceil\frac{c_p}{\eps^{1/2}}\log\frac{1}{\eps}\rceil$, where $c_p$ is the constant in Lemma~\ref{lem:proximity}. 

Note that the proximity result only works for a single capacity. When capacity changes, its break item may change as well, and as a consequence, a low-efficiency may become a high-efficiency item, and vice versa.

 We divide $I$ into $I_{\mathrm{head}} \cup I_{\mathrm{tail}}$ where $I_{\mathrm{head}} = \{i : 1 \leq i < b\}$ and $I_{\mathrm{tail}} = \{i : b \leq i \leq n\}$. Basically, $I_{\mathrm{tail}}$ is the set of items that can always be considered as low-efficiency items for all capacities up to $t$.

\subsection{Dealing with Items in the Tail} We shall further partition $I_{\mathrm{tail}}$. Let $b < i_0 < \cdots < i_{r} = n + 1$ be a sequence be obtained as follows.  For $j = 0, 1, 2, \ldots$, let $i_j$ be the largest index such that the items in $\{i : b \leq i < i_j\}$ have exactly $2^j\tau$ distinct profits. If no such index exists, we set $i_j = n + 1$ and ${r} = j$.  Then we partition $I_{\mathrm{tail}}$ into $I_{0} \cup I_1 \cup \cdots \cup I_{{r}}$ where $I_0 = \{i: b \leq i < i_0\}$ and $I_j = \{i: i_{j-1} \leq i < i_{j}\}$ for every integer $j \in [1, {r}]$.  See Figure~\ref{fig:tail} for an example.

\begin{figure}[ht]
        \centering
        \begin{tikzpicture}
        \draw [thick] (0,0) -- (8,0);
        \node [align=center, left] at (-0.2,0) {$I_{\mathrm{tail}}$};

        \draw [thick] (0, 0.1) -- (0, -0.1);
        \node [align =center, below] at (0, -0.2) {$b$};

        \draw [thick] (8, 0.1) -- (8, -0.1);
        \node [align =center, below] at (8, -0.2) {$n$};

        \draw [thick] (1.5, 0.1) -- (1.5, -0.1);
        \node [align =center, below] at (0.75, -0.2) {$I_0$};
        \node [align = center, rotate = -90, yscale = 3.5] at (0.75, 0.3) {$\{$};
        \node [align = center, above] at (0.75, 0.3) {\scriptsize {$\tau$ distinct profits}};

        \draw [thick] (3, 0.1) -- (3, -0.1);
        \node [align =center, below] at (2.25, -0.2) {$I_1$};
        \node [align = center, rotate = -90, yscale = 7.5] at (1.5, 1) {$\{$};
        \node [align = center, above] at (1.5, 1) {\scriptsize {$2\tau$ distinct profits}};

        \draw [thick] (6, 0.1) -- (6, -0.1);
        \node [align =center, below] at (4.5, -0.2) {$I_2$};
        \node [align = center, rotate = -90, yscale = 16] at (3, 1.7) {$\{$};
        \node [align = center, above] at (3, 1.7) {\scriptsize {$4\tau$ distinct profits}};

        \node [align =center, below] at (7, -0.2) {$\cdots$}; 

        \draw [thick, ->] (1, - 1) -- (7, - 1);
        \node [align = center, below] at (4,-1.2) {larger index implies lower efficiency};

        \end{tikzpicture}
        \caption{Partition of $I_{\mathrm{tail}}$}
        \label{fig:tail}
\end{figure}

    We make the following easy observations.
    \begin{itemize}
        \item ${r} \leq O(\log \frac{1}{\eps})$ since there are at most $\frac{1}{\eps}$ distinct profits.

        \item The items in $I_j$ have at most $2^{j}\tau$ distinct profits. 

        \item The partition $I_{0} \cup I_1 \cup \cdots \cup I_{{r}}$ can be computed in $O(n)$ time.
    \end{itemize}

Next we prove that for every $j \geq 1$, the subset $I_j$ can be considered as a set of low-efficiency items no matter what the capacity is. More precisely, for any capacity $x \in [0, t]$, there is an optimal solution $I^*$ for $(I, x)$ with small $p(I^* \cap I_j)$.

\begin{lemma}
    For every $x \in [0, t]$, there is an optimal solution $I^*$ for $(I, x)$ such that for every integer $j \in [1, {r}]$,
    \[
        p(I^* \cap I_j) \leq \frac{1}{2^{j-1}\eps^{3/2}}.
    \]
\end{lemma}
\begin{proof}
    Fix an arbitrary $x \in [0, t]$. Since $x \leq t$, the break item $b'$ for $x$ satisfies $b' \leq b$. By the construction of $I_j$, every $I_j$ is inferior to $b'$ by at least $2^{j-1}\tau$ distinct profits. By Lemma~\ref{lem:proximity} (with $p_{\max} = \frac{1}{\eps}$), there is an optimal solution $I^*$ for $(I, x)$ such that 
    \[
        p(I^* \cap I_j) \leq \frac{1}{2^{j-1}\eps^{3/2}}.
    \]
\end{proof}
We immediately have the following corollary. (Recall that $\mathrm{leq}(f, B)$ is the part of $f$ whose function values are at most $B$.)
\begin{corollary}\label{coro:tail-bound}
    $f_I = f_{I_{\mathrm{head}}} \oplus f_{I_0} \oplus \mathrm{leq}(f_{I_1}, \frac{1}{\eps^{3/2}}) \oplus \cdots \oplus \mathrm{leq}(f_{I_{r}}, \frac{1}{2^{{r}-1}\eps^{3/2}})$
\end{corollary}

Next we show that the functions $f_{I_0}$ and $\mathrm{leq}(f_{I_j}, \frac{1}{2^{j-1}\eps^{3/2}})$ can be approximated efficiently.

\begin{lemma}\label{lem:approx-I0}
    In $\widetilde{O}(|I_0| + \frac{1}{\eps^2})$ time, we can compute a function of complexity $\widetilde{O}(\frac{1}{\eps})$ that approximates $f_{I_0}$ with additive error $\widetilde{O}(\frac{1}{\alpha\eps})$.
\end{lemma}
\begin{proof}
    Note that $f_{I_0}(t) \leq f_I(t) \leq \frac{1}{\alpha\eps^2}$. To meet the additive error requirement, it suffices to approximate $f_{I_0}(t)$ with factor $1 + O(\eps)$.

    The items in $I_0$ have at most $\tau$ distinct profits. In $O(|I_0|)$ time, we can divide the items into (at most) $\tau$ groups $S_1, \ldots, S_{\tau}$ by profits. Since the items in each subset $S_k$ have the same profit, the corresponding function $f_{S_k}$ can computed in linear time. Moreover, each $f_{S_k}$ is pseudo-concave and $p$-uniform for some $p$. Thus we can use Lemma~\ref{lem:approx-m-conv}(ii) to compute a function of complexity $\widetilde{O}(\frac{1}{\eps})$ that approximate 
    \[
        f_{I_0} = f_{S_1} \oplus \cdots \oplus f_{S_{\tau}}
    \]
    with factor $1 + O(\eps)$. The time cost is $\widetilde{O}(\frac{1}{\eps}\tau^2) = \widetilde{O}(\frac{1}{\eps^2})$.
\end{proof}

\begin{lemma}\label{lem:approx-Ij}
    For each integer $j \in [1, {r}]$, we can compute $\mathrm{leq}(f_{I_j}, \frac{1}{2^{j-1}\eps^{3/2}})$ in time $\widetilde{O}(|I_j| + \frac{1}{\eps^2})$. Moreover, the complexity of $\mathrm{leq}(f_{I_j}, \frac{1}{2^{j-1}\eps^{3/2}})$ is $O(\frac{1}{2^{j-1}\eps^{3/2}})$.
\end{lemma}
\begin{proof}
    Fix an arbitrary integer $j \in [1, {r}]$. Recall that the items in $I_j$ have at most $2^j\tau$ distinct profits. In $O(|I_j|)$ time, we can divide the items into (at most) $2^j\tau$ groups $S_1, \ldots, S_{2^j\tau}$ by profits.  Since the items in each subset $S_k$ have the same profit, the corresponding function $f_{S_k}$ can computed in linear time. Each $f_{S_k}$ is pseudo-concave and $p$-uniform for some $p$. Moreover, function values of each $f_{S_k}$ are integers.  Thus, we can compute $f_{I_j} = f_{S_1} \oplus \cdots \oplus f_{S_{2^j\tau}}$ by invoking Lemma~\ref{lem:exact-two-conv} for $2^j\tau$ times. Note that our goal is to compute $\mathrm{leq}(f_{I_j}, \frac{1}{2^{j-1}\eps^{3/2}})$ rather than $f_{I_j}$. Every time before we invoke Lemma~\ref{lem:exact-two-conv} to compute $f\oplus g$, we can replace them with $\mathrm{leq}(f, \frac{1}{2^{j-1}\eps^{3/2}})$ and $\mathrm{leq}(g, \frac{1}{2^{j-1}\eps^{3/2}})$. Therefore, the time cost of each invocation of Lemma~\ref{lem:exact-two-conv} is $O(\frac{1}{2^{j-1}\eps^{3/2}})$. The total time cost is 
    \[
        O(|I_j|)+ O(\frac{1}{2^{j-1}\eps^{3/2}}) \cdot 2^j\tau =O(|I_j| + \frac{2\tau}{\eps^{3/2}}) = \widetilde{O}(|I_j| + \frac{1}{\eps^2}).
    \]
    Since all profits are integers, it is straightforward that $\mathrm{leq}(f_{I_j}, \frac{1}{2^{j-1}\eps^{3/2}})$ is of complexity $O(\frac{1}{2^{j-1}\eps^{3/2}})$.
\end{proof}

We summarize this section by the following lemma.
\begin{lemma}\label{lem:approx-tail}
    In $\widetilde{O}(n + \frac{1}{\eps^2})$ time, we can compute a function of complexity $\widetilde{O}(\frac{1}{\eps})$ 
    \[
        f_{I_0} \oplus \mathrm{leq}(f_{I_1}, \frac{1}{\eps^{3/2}}) \oplus \cdots \oplus \mathrm{leq}(f_{I_{r}}, \frac{1}{2^{{r}-1}\eps^{3/2}})
    \] 
    with additive error $\widetilde{O}(\frac{1}{\alpha\eps})$.
\end{lemma}
\begin{proof}
    Recall that $r \leq O(\log \frac{1}{\eps})$. By Lemma~\ref{lem:approx-I0} and Lemma~\ref{lem:approx-Ij}, in $\widetilde{O}(n + \frac{1}{\eps^2})$ time, we can compute all the functions $f_{I_0}, \mathrm{leq}(f_{I_1}, \frac{1}{\eps^{3/2}}), \ldots, \mathrm{leq}(f_{I_{r}}, \frac{1}{2^{{r}-1}\eps^{3/2}})$. The total complexity of these functions are $O(\frac{1}{\eps^{3/2}})$. Then we can approximately merge all these functions with factor $1 + O(\eps)$ using Lemma~\ref{lem:approx-m-conv}(i). The time cost is $\widetilde{O}(\frac{1}{\eps^2}r) = \widetilde{O}(\frac{1}{\eps^2})$. The additive error is at most $\widetilde{O}(\eps f_{I_{\mathrm{tail}}}(t)) \leq \widetilde{O}(\eps f_{I}(t)) = \widetilde{O}(\frac{1}{\alpha\eps})$
\end{proof}

\subsection{Dealing with Items in the Head}
It remains to handle $I_{\mathrm{head}}$. Recall that $I_{\mathrm{head}} = \{i : 1 \leq i < b\}$. Unlike the items in $I_{\mathrm{tail}}$, the items in $I_{\mathrm{head}}$ can be high-, or median-, or low-efficiency depending on what the capacity is. In order to apply the proximity result, we divide the interval $[0, t]$ into $\ell \leq \frac{1}{\alpha\eps^{1/2}}$ sub-intervals $[t_1, t_2] \cup [t_2, t_3] \cup \cdots \cup [t_{\ell}, t_{\ell+1}]$ in a way that for each sub-interval $[t_j, t_{j+1}]$, all capacities in this interval share a common set $I^j_H$ of high-efficiency items and a common set $I^j_L$ of low-efficiency items. Then the proximity result can be applied to each of these sub-intervals.

To obtain the sub-intervals, it suffices to determine the sequence $0 = t_1  < t_2 < \cdots < t_{\ell+1} = t$, and these capacities $t_j$ will be determined in a way that their break items $b_j$ satisfy certain properties. Therefore, it is more convenient to determine the break items $b_1 < b_2 < \cdots < b_{\ell+1}$ first. Let $b_1 = 1$. Given $b_j$, we let $b_{j+1}$ be the largest index such that the items in $\{i : b_j\leq i < b_{j+1}\}$ have exactly $\tau$ distinct profits. 
If no such $b_{j+1}$ exists, we set $b_{j+1} = b$. When $b_{j+1} = b$, we stop and set $\ell = j$. Then we determine the value of all capacities $t_j$. For $j \in [1, \ell]$, let $t_j = \sum_{i=1}^{b_j-1}w_i$. That is, $t_j$ is the smallest capacity whose break item is $b_j$. Let $t_{\ell + 1} = t$.  It is easy to see that the sequence $b_1 < \cdots < b_{\ell+1}$ and the sequence $t_1 < \cdots < t_{\ell + 1}$ can be obtained in $O(\frac{1}{\alpha\eps}) \leq O(\frac{1}{\eps})$ time by scanning the items in $I_{\mathrm{head}}$.

\begin{lemma}\label{lem:group-num}
    There are at most $O(\frac{1}{\alpha\eps^{1/2}})$ sub-intervals. That is, $\ell \leq O(\frac{1}{\alpha\eps^{1/2}})$.
\end{lemma}   
\begin{proof}
    By the construction of the break items, for any integer $j \in [1, \ell-1]$,  we have $b_{j+1} \geq b_j + \tau$. It implies that $b_{\ell} \geq b_1 + (\ell - 1)\tau$. Note that $b_{\ell} < b_{\ell+1} = b$. Recall that $b \leq \frac{2}{\alpha\eps} + 1$. Therefore, $b_{\ell} \leq \frac{2}{\alpha\eps}$. Now we have 
    \[
       b_1 + (\ell - 1)\tau \leq  b_{\ell} \leq \frac{2}{\alpha\eps}.
    \]
    It is easy to see that $\ell \leq O(\frac{1}{\alpha\eps\tau}) = O(\frac{1}{\alpha\eps^{1/2}})$.
\end{proof}

\begin{lemma}\label{lem:break-item-x}
    For any integer $j \in [1, \ell]$, for any $x \in [t_{j}, t_{j+1}]$, the break item for $x$ lies in $\{i : b_j \leq i \leq b_{j+1}\}$.
\end{lemma}
\begin{proof}
    By the definition of $t_j$, the break item of $t_j$ is $b_j$. It is straightforward that the break item of $x$ lies between $b_j$ and $b_{j+1}$.
\end{proof}

Next we shall define the set of high-efficiency items and the set of low-efficiency items for each sub-interval. For simplicity, we define $b_{-1} = b_0 = b_1$ and $b_{\ell+2} = b$. For each integer $j \in [1, \ell]$, we define $I^j_H = \{i : 1 \leq i < b_{j-1}\}$, $I^j_M = \{i: b_{j-1} \leq i <  b_{j+2}\}$, and $I^j_L = \{i: b_{j+2} \leq i < b\}$. See Figure~\ref{fig:head} for an example.

\begin{figure}[ht]
        \centering
        \begin{tikzpicture}
        \draw [thick] (0,0) -- (10.5,0);
        \node [align=center, left] at (-0.2,0) {$I_{\mathrm{tail}}$};

        \draw [thick] (0, 0.1) -- (0, -0.1);
        \node [align =center, below] at (0, -0.2) {$b_1$};

        \draw [thick] (10.5, 0.1) -- (10.5, -0.1);
        \node [align =center, below] at (10.5, -0.2) {$b - 1$};

        \node [below] at (1.5, -0.1) {$\cdots$};
        \node [below] at (9, -0.1) {$\cdots$};

        \draw [thick] (3, 0.1) -- (3, -0.1);
        \node [align =center, below] at (3, -0.2) {$b_{j-1}$};
        \node [align = center, rotate = -90, yscale = 3.5] at (3.75, 0.3) {$\{$};
        \node [align = center, above] at (3.75, 0.3) {\scriptsize {$\tau$ distinct}\\{\scriptsize profits}};

        \draw [thick] (4.5, 0.1) -- (4.5, -0.1);
        \node [align =center, below] at (4.5, -0.2) {$b_{j}$};
        \node [align = center, rotate = -90, yscale = 3.5] at (5.25, 0.3) {$\{$};
        \node [align = center, above] at (5.25, 0.3) {\scriptsize {$\tau$ distinct}\\{\scriptsize profits}};

        \draw [thick] (6, 0.1) -- (6, -0.1);
        \node [align =center, below] at (6, -0.2) {$b_{j+1}$};
        \node [align = center, rotate = -90, yscale = 3.5] at (6.75, 0.3) {$\{$};
        \node [align = center, above] at (6.75, 0.3) {\scriptsize {$\tau$ distinct}\\{\scriptsize profits}};

        \draw [thick] (7.5, 0.1) -- (7.5, -0.1);
        \node [align =center, below] at (7.5, -0.2) {$b_{j+2}$};

        \node [align = center, rotate = 90, yscale = 8] at (1.5, -0.9) {$\{$};
        \node [align = center, below] at (1.5, -0.9) {$I^j_H$};

        \node [align = center, rotate = 90, yscale = 8] at (9, -0.9) {$\{$};
        \node [align = center, below] at (9, -0.9) {$I^j_L$};

        \node [align = center, rotate = 90, yscale = 13] at (5.25, -0.9) {$\{$};
        \node [align = center, below] at (5.25, -0.9) {$I^j_M$};

        \draw [thick, ->] (1.5, - 2) -- (9, - 2);
        \node [align = center, below] at (5.25,-2.2) {larger index implies lower efficiency};

        \end{tikzpicture}
        \caption{A Partition of $I_{\mathrm{head}}$ for $[t_j, t_{j+1}]$}
        \label{fig:head}
\end{figure}

\begin{lemma}
    For any integer $j \in [1, \ell]$ and any $x \in [t_j, t_{j+1}]$, there is an optimal solution $I^*$ for Knapsack instance $ (I_{\mathrm{head}}, x)$ such that 
        \[
            p(I^*\cap I^j_H) \geq p(I^j_H) - \frac{1}{\eps^{3/2}} \quad\text{and}\quad p(I^*\cap I^j_L) \leq \frac{1}{\eps^{3/2}}.
        \] 
\end{lemma}
\begin{proof}
    It is easier to understand this proof using Figure~\ref{fig:head}. By Lemma~\ref{lem:break-item-x}, the break item $b'$ for $x$ lies between $b_j$ and $b_{j+1}$. Recall that $I^j_H = \{i : 1 \leq i < b_{j-1}\}$ and the items in $\{i : b_{j-1} \leq i < b_j\}$ have $\tau$ distinct profit. Therefore, $I^j_H$ is superior to $b'$ by at least $\tau$ distinct profits. Similarly, one can also verify that $I^j_L$ is inferior to $b'$ by at least $\tau$ distinct profits. By Lemma~\ref{lem:proximity}, there is an optimal solution $I^*$ for $(I_{\mathrm{head}}, x)$ such that 
        \[
            p(I^*\cap I^j_H) \geq p(I^j_H) - \frac{1}{\eps^{3/2}} \quad\text{and}\quad p(I^*\cap I^j_L) \leq \frac{1}{\eps^{3/2}}.
        \] 
\end{proof}
We immediately have the following corollary. (Recall that $\mathrm{geq}(f, B)$ represents the part of $f$ whose function values are at least $B$.)
\begin{corollary}
    For every integer $j \in [1, \ell]$, $f_{I_\mathrm{head}}$ and the following function
    \[
        \mathrm{geq}(f_{I^j_H}, p(I^j_H) - \frac{1}{\eps^{3/2}}) \oplus f_{I^j_M} \oplus \mathrm{leq}(f_{I^j_L},\frac{1}{\eps^{3/2}})
    \] 
    have the same function value on any $x \in [t_j, t_{j+1}]$.
\end{corollary}
Note that for $x \notin [t_j, t_{j+1}]$, the function value of $f_{I_\mathrm{head}}$ is greater than or equal to that of $\mathrm{geq}(f_{I^j_H}, p(I^j_H) - \frac{1}{\eps^{3/2}}) \oplus f_{I^j_M} \oplus \mathrm{leq}(f_{I^j_L},\frac{1}{\eps^{3/2}})$. Therefore, we have the following corollary.
\begin{corollary}\label{coro:comb-int}
    $f_{I_\mathrm{head}} = \max\{g_1, g_2, \ldots, g_\ell\}$ where
    \[
        g_j = \mathrm{geq}(f_{I^j_H}, p(I^j_H) - \frac{1}{\eps^{3/2}}) \oplus f_{I^j_M} \oplus \mathrm{leq}(f_{I^j_L},\frac{1}{\eps^{3/2}}).
    \]
\end{corollary}

The groups of median-efficiency items also have some nice properties, which will be used later to obtain an efficient approximation of these groups. 
\begin{lemma}\label{lem:median-property}
    For any integer $j \in [1,\ell]$, the items in $I^j_M$ have at most $3\tau$ distinct profits. Moreover, $\sum_{j = 1}^{\ell} |I^j_M| \leq \frac{6}{\alpha\eps}$.
\end{lemma}
\begin{proof}
    For each $j \in \{1, \ldots, \ell\}$, by the definition of $b_j$, we have that the items in $\{i : b_j \leq i < b_{j+1}\}$ have at most $\tau$ distinct profits. Note that 
    \begin{align*}
        I^j_M &= \{i: b_{j-1} \leq i < b_{j+2}\}\\
            & = \{i: b_{j-1} \leq  i < b_j\} \cup \{i: b_{j} \leq i < b_{j+1}\} \cup \{i: b_{j+1} \leq i < b_{j+2}\}.
    \end{align*}
    Therefore, the items in $I^j_M$ have at most $3\tau$ distinct profits.

    It is easy to see that $I^j_M \subseteq I_{\mathrm{head}}$. Although the groups $I_j^M$ may overlap, by definition of $I^j_M$, every item in $I_{\mathrm{head}}$ appears in at most three of them. Therefore, 
    \[
        \sum_{j = 1}^{\ell} |I^j_M| \leq 3|I_{\mathrm{head}}| \leq 3(b - 1) \leq \frac{6}{\alpha\eps}.
    \]
    The last inequality is due to $b \leq \frac{2}{\alpha\eps} + 1$.
\end{proof}

By Corollary~\ref{coro:comb-int}, to compute $f_{I_{\mathrm{head}}}$, it suffices to compute the functions corresponding to $I^j_H$, $I^j_M$, and $I^j_L$, and then merge them to obtain $g_j$.

\subsubsection{Dealing with Groups of Low-Efficiency Items}
We will show that all functions $\mathrm{leq}(f_{I^j_L},\frac{1}{\eps^{3/2}})$ can be approximately computed in $\widetilde{O}(n + \frac{1}{\eps^2})$ total time.

\begin{lemma}\label{lem:low-group}
    In $O(\frac{1}{\eps^2})$ total time, we can compute a function of complexity $O(\frac{\alpha}{\eps})$ that approximates $\mathrm{leq}(f_{I^j_L},\frac{1}{\eps^{3/2}})$ with additive error $O(\frac{1}{\alpha\eps})$ for each integer $j \in [1, \ell]$. 
\end{lemma}
\begin{proof}
    The maximum function value is $\frac{1}{\eps^{3/2}}$ and the permissible additive error is $O(\frac{1}{\alpha\eps})$, so it suffices to approximate the function with factor $1 + O(\frac{\eps^{1/2}}{\alpha})$. Recall that $p_i \in [\frac{1}{2\eps}, \frac{1}{\eps}]$. We multiply each $p_i$ by $\alpha\eps^{1/2}$ for every every $i \in I_{\mathrm{head}}$, and this makes $p_i \in [\frac{\alpha}{2\eps^{1/2}}, \frac{\alpha}{\eps^{1/2}}]$ and the maximum function value be $\frac{\alpha}{\eps}$. Let $B = \frac{\alpha}{\eps}$. We round each $p_i$ down to the nearest integer. This changes the function value by a factor of at most $1 + O(\frac{\eps^{1/2}}{\alpha})$.

    We compute $\mathrm{leq}(f_{I^j_L},B)$ for $j = \ell, \ell - 1, \ldots, 2, 1$ using standard dynamic programming. The function $\mathrm{leq}(f_{I^\ell_L},B)$ can be computed in $O(|I^\ell_L| B)$ time. Suppose that we have computed $\mathrm{leq}(f_{I^j_L},B)$. Note that $I^{j}_L = \{i : b_{j+2} \leq i < b \}$ is a subset of $I^{j-1}_L = \{i: b_{j+1} \leq i < b\}$. Therefore, the dynamic programming result for $I^j_L$ can be extended to a result for $I^{j-1}_L$ in $O((|I^{j-1}_L| - |I^j_L|)B)$ time. In other words, the function $\mathrm{leq}(f_{I^{j-1}_L},B)$ can be computed from $\mathrm{leq}(f_{I^j_L},B)$ in $O((|I^{j-1}_L| - |I^j_L|)B)$ time. The total time cost for computing all these functions is therefore
    \[
        O(|I^1_L| B) \leq O(|I_\mathrm{head}| B) \leq O((b-1)\frac{\alpha}{\eps}) \leq O(\frac{1}{\eps^2}).
    \]
    The last inequality is due to that $b \leq \frac{2}{\alpha\eps} + 1$.

    Since every $p_i$ is an integer, each function $\mathrm{leq}(f_{I^j_L},B)$ is of complexity $O(B) = O(\frac{\alpha}{\eps})$.
\end{proof}

\subsubsection{Dealing with Groups of High-Efficiency Items}
The idea for handling the functions $\mathrm{geq}(f_{I^j_H}, p(I^j_H) - \frac{1}{\eps^{3/2}})$ is similar to that for $\mathrm{leq}(f_{I^j_L},\frac{1}{\eps^{3/2}})$. The only difference is that we consider a problem ``complementary'' to Knapsack: instead of determining which items should be selected, this problem determines which items should not be selected. 

\begin{lemma}\label{lem:hight-group}
    In $O(\frac{1}{\eps^2})$ total time, we can compute a function of complexity $O(\frac{\alpha}{\eps})$ that approximates $\mathrm{geq}(f_{I^j_H}, p(I^j_H) - \frac{1}{\eps^{3/2}})$ with additive error $O(\frac{1}{\alpha\eps})$ for each integer $j \in [1, \ell]$. 
\end{lemma}
\begin{proof}
    Define
    \[
        h_{{I}^j_H}(x) = \min \left\{\textrm{$p(I')$ : $I' \subseteq I^j_H$ and $w(I') \geq t$} \right\}. 
    \]
    for $x \in \mathbb{R}$. It is easy to observe that for any $x \in [0, t]$,
    \[
        f_{{I}^j_H}(x) = p(I^j_H) - h_{{I}^j_H}(w({I}^j_H) - x).
    \]
    To approximate $\mathrm{geq}(f_{I^j_H}, p(I^j_H) - \frac{1}{\eps^{3/2}})$ with additive error $O(\frac{1}{\alpha\eps})$, it suffices to approximate $\mathrm{leq}(H_{I^j_H},\frac{1}{\eps^{3/2}})$ with the same error. (We remark that to approximate the functions $f_{I^j_H}$ from below, we should approximate $h_{{I}^j_H}(x)$ from above. Therefore, when approximating $h_{{I}^j_H}(x)$, the profit will be rounded up rather than rounded down.) We can handle the functions $\mathrm{leq}(H_{I^j_H},\frac{1}{\eps^{3/2}})$ in the same way as we handle the functions $\mathrm{leq}(f_{I^j_L},\frac{1}{\eps^{3/2}})$. We provide the details for completeness.

    The maximum function value of $\mathrm{leq}(h_{I^j_H},\frac{1}{\eps^{3/2}})$ is $\frac{1}{\eps^{3/2}}$ and the permissible additive error is $O(\frac{1}{\alpha\eps})$, so it suffices to approximate the function with factor $1 + O(\frac{\eps^{1/2}}{\alpha})$. Recall that $p_i \in [\frac{1}{2\eps}, \frac{1}{\eps}]$. We multiply each $p_i$ by $\alpha\eps^{1/2}$ for every every $i \in I_{\mathrm{head}}$, and this makes $p_i \in [\frac{\alpha}{2\eps^{1/2}}, \frac{\alpha}{\eps^{1/2}}]$ and the maximum function value be $\frac{\alpha}{\eps}$. Let $B = \frac{\alpha}{\eps}$. We round each $p_i$ up to the nearest integer.  This changes the function value by a factor of at most $1 + O(\frac{\eps^{1/2}}{\alpha})$. 

    We compute $\mathrm{leq}(h_{I^j_H}, B)$ for $j = 1,2,\ldots, \ell$ using standard dynamic programming. The function $\mathrm{leq}(h_{I^1_H},B)$ can be computed in $O(|I^1_H| B)$ time. Suppose that we have computed $\mathrm{leq}(h_{I^j_H},B)$. Note that $I^{j}_H = \{i : 1 \leq i < b_{j-1} \}$ is a subset of $I^{j+1}_H = \{i: 1 \leq i < b_j\}$. Therefore, the dynamic programming result for $I^j_H$ can be extended to a result for $I^{j+1}_H$ in $O((|I^{j+1}_H| - |I^j_H|)B)$ time. In other words, the function $\mathrm{leq}(h_{I^{j+1}_H},B)$ can be computed from $\mathrm{leq}(h_{I^j_H},B)$ in $O((|I^{j+1}_H| - |I^j_H|)B)$ time. The total time cost for computing all these functions is therefore
    \[
        O(|I^\ell_H| B) \leq O(|I_\mathrm{head}| B) \leq O((b-1)\frac{\alpha}{\eps}) \leq O(\frac{1}{\eps^2}).
    \]
    The last inequality is due to that $b \leq \frac{2}{\alpha\eps} + 1$.

    Since every $p_i$ is an integer, each function $\mathrm{leq}(h_{I^j_H},B)$ is of complexity $O(B) = O(\frac{\alpha}{\eps})$, so is the function $\mathrm{geq}(f_{I^j_H}, p(I^j_H) - \frac{1}{\eps^{3/2}})$.
\end{proof}

\subsubsection{Dealing with Groups of Median-Efficiency Items}
It is tempting to compute all functions $f_{I^j_M}$, and the merge them with the functions $\mathrm{geq}(f_{I^j_H}, p(I^j_H) - \frac{1}{\eps^{3/2}})$ and $\mathrm{leq}(f_{I^j_L},\frac{1}{\eps^{3/2}})$. This, however, may lead to a super-quadratic running time, since merging two functions needs $O(\frac{1}{\eps^2})$ time in general and we need to merge $O(\ell) = O(\frac{1}{\alpha\eps^{1/2}})$ time.  To obtain a quadratic running time, we should fully utilize the following two properties of these functions: (i) $\mathrm{geq}(f_{I^j_H}, p(I^j_H) - \frac{1}{\eps^{3/2}})$ and $\mathrm{leq}(f_{I^j_L},\frac{1}{\eps^{3/2}})$ have a small range of function values; (ii) items in $I^j_M$ have a small number of distinct profits.

We first merge $\mathrm{geq}(f_{I^j_H}, p(I^j_H) - \frac{1}{\eps^{3/2}})$ and $\mathrm{leq}(f_{I^j_L},\frac{1}{\eps^{3/2}})$. 
\begin{lemma}\label{lem:merge-low-high}
    In $\widetilde{O}(\frac{1}{\eps^2})$ total time, we can compute a function of $\widetilde{O}(\frac{\alpha}{\eps^{1/2}})$ complexity that approximates 
    \[
        \mathrm{geq}(f_{I^j_H}, p(I^j_H) - \frac{1}{\eps^{3/2}}) \oplus \mathrm{leq}(f_{I^j_L},\frac{1}{\eps^{3/2}})  
    \] 
    with additive error $O(\frac{1}{\alpha\eps})$ for each integer $j \in [1, \ell]$.
\end{lemma}
\begin{proof}
    Let $f^1_L, \ldots, f^\ell_L$ be the $\ell$ functions given by Lemma~\ref{lem:low-group}. Let $f^1_H, \ldots, f^\ell_H$ be the $\ell$ functions we computed by Lemma~\ref{lem:hight-group}. Computing all these $2\ell$ functions takes $O(\frac{1}{\eps^2})$ time. Each of these functions is of complexity $O(\frac{\alpha}{\eps})$.

    It remains to compute $f^j_L \oplus f^j_H$ for every integer $j \in [1, \ell]$. By Lemma~\ref{lem:group-num}, $\ell \leq O(\frac{1}{\alpha\eps^{1/2}})$. Recall that $\alpha \in [1, \frac{1}{\eps}]$.  If $\alpha \geq \frac{1}{\eps^{1/2}}$, then $\ell = O(1)$. We can compute a function of complexity $\widetilde{O}(\frac{1}{\eps}) = \widetilde{O}(\frac{\alpha}{\eps^{1/2}})$ that approximate $f^j_L \oplus f^j_H$ with factor $1 + O(\eps)$ in $\widetilde{O}(\frac{\alpha}{\eps} + \frac{1}{\eps^2}) = \widetilde{O}(\frac{1}{\eps^2})$ time using Lemma~\ref{lem:approx-m-conv}(i). Note that the additive error is at most $O(\eps\cdot p(I_{\mathrm{head}})) = O(\frac{1}{\alpha\eps})$. 

    Suppose that $\alpha < \frac{1}{\eps^{1/2}}$. Note that the function values of $f^j_L$ are contained in $[0, \frac{1}{\eps^{3/2}}]$ and the function values of $f^j_H$ are contained in $[p(I^j_H) - \frac{1}{\eps^{3/2}}, p(I^j_H)]$.  Let $g^j_L = f^j_L + \frac{1}{\eps^{3/2}}$ and let $g^j_H = f^j_H - p(I^j_H) + \frac{2}{\eps^{3/2}}$. To compute $f^j_H \oplus f^j_L$, it suffices to compute $g^j_L \oplus g^j_H$.  Note that the function values of $g^j_L$ and $g^j_H$ are contained in $[\frac{1}{\eps^{3/2}}, \frac{2}{\eps^{3/2}}]$. Meanwhile, the permissible additive error is $O(\frac{1}{\alpha\eps})$. Therefore, it suffices to approximate $g^j_L \oplus g^j_H$ with factor $1 + O(\frac{\eps^{1/2}}{\alpha})$.  By Lemma~\ref{lem:approx-m-conv}(i), in $\widetilde{O}(\frac{\alpha}{\eps} + \frac{\alpha^2}{\eps}) = \widetilde{O}( \frac{\alpha^2}{\eps})$ time, we can compute a function of complexity $\widetilde{O}(\frac{\alpha}{\eps^{1/2}})$ that approximates $g^j_L \oplus g^j_H$ with factor $1 + O(\frac{\eps^{1/2}}{\alpha})$.  Therefore, the total time cost to approximately compute all functions $g^j_L \oplus g^j_H$ (and therefore all functions $f^j_L \oplus f^j_H$) is
    \[
          \widetilde{O}( \frac{\alpha^2}{\eps}) \cdot \ell \leq \widetilde{O}(\frac{\alpha^2}{\eps} \cdot \frac{1}{\alpha\eps^{1/2}}) = \widetilde{O}(\frac{\alpha}{\eps^{3/2}}) \leq \widetilde{O}(\frac{1}{\eps^2}).
    \]
    The last inequality is due to that $\alpha < \frac{1}{\eps^{1/2}}$.
\end{proof}

Next, we consider $I^j_M$. Instead of computing $f_{I^j_M}$, we divide $I^j_M$ into groups by profit, compute the functions for each group, and merge it into the resulting functions of Lemma~\ref{lem:merge-low-high}.

\begin{lemma}\label{lem:approx-lmh}
    In $\widetilde{O}(\frac{1}{\eps^2})$ total time, we can compute a function that approximates 
    \[
        \mathrm{geq}(f_{I^j_H}, p(I^j_H) - \frac{1}{\eps^{3/2}}) \oplus f_{I^j_M} \oplus \mathrm{leq}(f_{I^j_L},\frac{1}{\eps^{3/2}})  
    \] 
    with additive error $O(\frac{1}{\alpha\eps})$ for each integer $j \in [1, \ell]$. Moreover, the total complexity of these functions is $\widetilde{O}(\frac{1}{\eps})$.
\end{lemma}
\begin{proof}
    Let $f_1, \ldots, f_\ell$ be the $\ell$ functions given by Lemma~\ref{lem:merge-low-high}. These functions can be computed in $\widetilde{O}(\frac{1}{\eps^2})$ time, and each is of complexity $\widetilde{O}(\frac{\alpha}{\eps^{1/2}})$.

    It remains to compute $f_j \oplus f_{I^j_M}$ for every integer $j \in [1, \ell]$.  Note that the function values of $f_j$ are contained in $[p(I^j_H) - \frac{1}{\eps^{3/2}}, p(I^j_H) + \frac{1}{\eps^{3/2}}]$. Let $g_j = f_j - p(I^j_H) + \frac{2}{\eps^{3/2}}$. To compute $f_j \oplus f_{I^j_M}$, it suffices to compute $g_j \oplus f_{I^j_M}$. Note that the function values of $g_j$ are contained in $[\frac{1}{\eps^{3/2}}, \frac{3}{\eps^{3/2}}]$. Let $s_j = |I^j_M|$. The function values of $f_{I^j_M}$ are contained in $\{0\} \cup [\frac{1}{2\eps}, \frac{s_j}{\eps}]$.  Therefore, the function values of $g_j \oplus f_{I^j_M}$ are contained in $[1, \frac{s_j}{\eps} + \frac{3}{\eps^{3/2}}]$. Recall that $s_j \geq \tau \geq \frac{1}{\eps^{1/2}}$. Therefore, the function values of $g_j \oplus f_{I^j_M}$ are contained in $[1, O(\frac{s_j}{\eps})]$. Since the permissible additive error is $\widetilde{O}(\frac{1}{\alpha\eps})$, it suffices to approximate $g_j \oplus f_{I^j_M}$ with factor $1 + O(\frac{1}{\alpha s_j})$.

    By Lemma~\ref{lem:median-property}, the items in $I^j_M$ have at most $3\tau$ distinct profits. In $O(s_j)$ time, we can partition $I^j_M$ into (at most) $3\tau$ subsets $S_1, \ldots, S_{3\tau}$ by profit. Since the items in each subset $S_j$ have the same profit, the corresponding function $f_{S_j}$ can computed in linear time, and moreover, it is pseudo-concave and $p$-uniform for some $p$. Thus we can use Lemma~\ref{lem:approx-m-conv}(ii) to compute a function of complexity $\widetilde{O}(\alpha s_j)$ that approximates 
    \[
        f_j \oplus f_{S_1} \oplus \cdots \oplus f_{S_{3\tau}}
    \]
    with factor $1 + O(\frac{1}{\alpha s_j})$. The time cost is $\widetilde{O}(\frac{\alpha}{\eps^{1/2}} + s_j + \alpha s_j\tau^2) = \widetilde{O}(\alpha s_j\tau^2)$. Summing the complexity of the resulting functions and the time cost over all $j$'s, we have the total complexity of the resulting functions is 
    \[
        \sum_{j=1}^\ell \widetilde{O}(\alpha s_j) \leq \widetilde{O}(\alpha \sum_{j=1}^\ell s_j) \leq \widetilde{O}(\alpha \cdot \frac{1}{\alpha\eps}) \leq  \widetilde{O}(\frac{1}{\eps}),
    \]
    where the second inequality is due to Lemma~\ref{lem:median-property}, the and the total time cost is
    \[
        \sum_{j=1}^\ell \widetilde{O}(\alpha s_j\tau^2) \leq  \widetilde{O}(\alpha \tau^2 \sum_{j=1}^\ell s_j) \leq \widetilde{O}(\frac{\tau^2}{\eps}) \leq \widetilde{O}(\frac{1}{\eps^2}).
    \]
\end{proof}

The following lemma is a consequence of Corollary~\ref{coro:comb-int} and Lemma~\ref{lem:approx-lmh}.
\begin{lemma}\label{lem:head}
    In $\widetilde{O}(\frac{1}{\eps^2})$ total time, we can compute a function of complexity $\widetilde{O}(\frac{1}{\eps^2})$ that approximates $f_{I_{\mathrm{head}}}$
    with additive error $O(\frac{1}{\alpha\eps})$
\end{lemma}

\subsection{Putting Things Together}
\begin{lemma}\label{lem:solve-rp}
    The reduce problem $\mathrm{RP}(\eps, \alpha)$ can be solved in $\widetilde{O}(n + \frac{1}{\eps^2})$ time.
\end{lemma}
\begin{proof}
    By Lemma~\ref{lem:head}, we can compute a function of complexity $\tilde{O}(\frac{1}{\eps})$ that approximates $f_{I_{\mathrm{head}}}$ with additive error $\widetilde{O}(\frac{1}{\alpha\eps})$.

    By Lemma~\ref{lem:approx-tail}, in $\widetilde{O}(n + \frac{1}{\eps^2})$ time, we can compute a function of complexity $\widetilde{O}(\frac{1}{\eps})$ 
    \[
        f_{I_0} \oplus \mathrm{leq}(f_{I_1}, \frac{1}{\eps^{3/2}}) \oplus \cdots \oplus \mathrm{leq}(f_{I_{r}}, \frac{1}{2^{{r}-1}\eps^{3/2}})
    \] 
    with additive error $\widetilde{O}(\frac{1}{\alpha\eps})$.

    Then we merge the two resulting functions by Lemma~\ref{lem:approx-m-conv}(i) (with factor $1 + O(\eps))$. This incurs an additional additive error $O(\eps f_I(t)) = O(\frac{1}{\alpha\eps})$. Therefore, we can obtain a function of complexity $\widetilde{O}(\frac{1}{\eps})$ that, by Corollary~\ref{coro:tail-bound}, approximates $f_I$ with additive error $\widetilde{O}(\frac{1}{\alpha\eps})$.
\end{proof}

Theorem~\ref{thm:main} follows by Lemma~\ref{lem:reduction} and Lemma~\ref{lem:solve-rp}.

\section{Conclusion}\label{sec:conclusion}
In this paper, we presented an FPTAS with running time $\widetilde{O}(n + (1/\eps)^2)$ for Knapsack, which is essentially tight given the conditional lower bound of $(n+1/\eps)^{2-o(1)}$. It remains an important open problem whether there exists an FPTAS of running time $\widetilde{O}(n/\eps)$ for Knapsack. Another open problem is that if we are allowed to relax the knapsack constraint by $O(\eps)$ fraction, can we break the quadratic barrier? This concept is also known as weak approximation schemes, and has been studied extensively for Subset 
Sum~\cite{MWW19,BN21b,CLMZ24cSTOCPartition} and Unbounded Knapsack recently~\cite{KPS17,JK18,CMWW19,MWW19,BC22}.

\section*{Acknowledgements} We are grateful to the anonymous reviewers for their valuable comments, which helped us significantly simplify the presentation of our algorithm.

\bibliographystyle{alphaurl}
\bibliography{newref}

\begin{thebibliography}{CMWW19}

\bibitem[AT19]{AT19}
Kyriakos Axiotis and Christos Tzamos.
\newblock Capacitated {{Dynamic Programming}}: {{Faster Knapsack}} and {{Graph Algorithms}}.
\newblock In {\em 46th {{International Colloquium}} on {{Automata}}, {{Languages}}, and {{Programming}}}, pages 19:1--19:13. Schloss Dagstuhl -- Leibniz-Zentrum f{\"u}r Informatik, 2019.

\bibitem[BC22]{BC22}
Karl Bringmann and Alejandro Cassis.
\newblock Faster {{Knapsack Algorithms}} via {{Bounded Monotone Min-Plus-Convolution}}.
\newblock In {\em 49th {{International Colloquium}} on {{Automata}}, {{Languages}}, and {{Programming}}}, pages 31:1--31:21. Schloss Dagstuhl -- Leibniz-Zentrum f{\"u}r Informatik, 2022.

\bibitem[BN21]{BN21b}
Karl Bringmann and Vasileios Nakos.
\newblock A {{Fine-Grained Perspective}} on {{Approximating Subset Sum}} and {{Partition}}.
\newblock In {\em Proceedings of the 2021 {{ACM-SIAM Symposium}} on {{Discrete Algorithms}}}, pages 1797--1815. Society for Industrial and Applied Mathematics, 2021.

\bibitem[Bri24]{Bri24}
Karl Bringmann.
\newblock Knapsack with small items in near-quadratic time.
\newblock In {\em Proceedings of the 56th Annual ACM Symposium on Theory of Computing}, pages 259--270. Association for Computing Machinery, 2024.

\bibitem[BW21]{BW21}
Karl Bringmann and Philip Wellnitz.
\newblock On {{Near-Linear-Time Algorithms}} for {{Dense Subset Sum}}.
\newblock In {\em Proceedings of the 2021 {{ACM-SIAM Symposium}} on {{Discrete Algorithms}}}, pages 1777--1796. Society for Industrial and Applied Mathematics, 2021.

\bibitem[CFG89]{CFG89}
Mark Chaimovich, Gregory Freiman, and Zvi Galil.
\newblock Solving dense subset-sum problems by using analytical number theory.
\newblock {\em J. Complexity}, 5(3):271--282, September 1989.
\newblock \href {https://doi.org/10.1016/0885-064X(89)90025-3} {\path{doi:10.1016/0885-064X(89)90025-3}}.

\bibitem[CFP21]{CFP21}
David Conlon, Jacob Fox, and Huy~Tuan Pham.
\newblock Subset sums, completeness and colorings, April 2021.

\bibitem[Cha99]{Cha99}
Mark Chaimovich.
\newblock New algorithm for dense subset-sum problem.
\newblock In Deshouilliers Jean-Marc, Landreau Bernard, and Yudin~Alexander A., editors, {\em Structure theory of set addition}, number 258 in Ast\'erisque, pages 363--373. Soci\'et\'e math\'ematique de France, 1999.

\bibitem[Cha18]{Chan18}
Timothy~M. Chan.
\newblock Approximation {{Schemes}} for 0-1 {{Knapsack}}.
\newblock In {\em 1st {{Symposium}} on {{Simplicity}} in {{Algorithms}}}, pages 5:1--5:12. Schloss Dagstuhl -- Leibniz-Zentrum f{\"u}r Informatik, 2018.

\bibitem[CLMZ24a]{CLMZ24cSTOCPartition}
Lin Chen, Jiayi Lian, Yuchen Mao, and Guochuan Zhang.
\newblock Approximating partition in near-linear time.
\newblock In {\em Proceedings of the 56th Annual ACM Symposium on Theory of Computing}, pages 307--318. Association for Computing Machinery, 2024.

\bibitem[CLMZ24b]{CLMZ24aSODA}
Lin Chen, Jiayi Lian, Yuchen Mao, and Guochuan Zhang.
\newblock Faster algorithms for bounded knapsack and bounded subset sum via fine-grained proximity results.
\newblock In {\em Proceedings of the 2024 Annual ACM-SIAM Symposium on Discrete Algorithms}, pages 4828--4848. Society for Industrial and Applied Mathematics, 2024.

\bibitem[CMWW19]{CMWW19}
Marek Cygan, Marcin Mucha, Karol W{\k{e}}grzycki, and Micha{\l} W{\l}odarczyk.
\newblock On problems equivalent to (min,+)-convolution.
\newblock {\em ACM Trans. Algorithms}, 15(1):1--25, January 2019.
\newblock \href {https://doi.org/10.1145/3293465} {\path{doi:10.1145/3293465}}.

\bibitem[DJM23]{DJM23}
Mingyang Deng, Ce~Jin, and Xiao Mao.
\newblock Approximating {{Knapsack}} and {{Partition}} via {{Dense Subset Sums}}.
\newblock In {\em Proceedings of the 2023 {{Annual ACM-SIAM Symposium}} on {{Discrete Algorithms}}}, pages 2961--2979. Society for Industrial and Applied Mathematics, 2023.

\bibitem[EW19]{EW19}
Friedrich Eisenbrand and Robert Weismantel.
\newblock Proximity {{Results}} and {{Faster Algorithms}} for {{Integer Programming Using}} the {{Steinitz Lemma}}.
\newblock {\em ACM Trans. Algorithms}, 16(1):5:1--5:14, November 2019.
\newblock \href {https://doi.org/10.1145/3340322} {\path{doi:10.1145/3340322}}.

\bibitem[Fre93]{Fre93}
Gregory~A Freiman.
\newblock New analytical results in subset-sum problem.
\newblock {\em Discrete Math.}, 114(1):205--217, April 1993.
\newblock \href {https://doi.org/10.1016/0012-365X(93)90367-3} {\path{doi:10.1016/0012-365X(93)90367-3}}.

\bibitem[GM91]{GM91}
Zvi Galil and Oded Margalit.
\newblock An {{Almost Linear-Time Algorithm}} for the {{Dense Subset-Sum Problem}}.
\newblock {\em SIAM J. Comput.}, 20(6):1157--1189, December 1991.
\newblock \href {https://doi.org/10.1137/0220072} {\path{doi:10.1137/0220072}}.

\bibitem[IK75]{IK75}
Oscar~H. Ibarra and Chul~E. Kim.
\newblock Fast {{Approximation Algorithms}} for the {{Knapsack}} and {{Sum}} of {{Subset Problems}}.
\newblock {\em J. ACM}, 22(4):463--468, October 1975.
\newblock \href {https://doi.org/10.1145/321906.321909} {\path{doi:10.1145/321906.321909}}.

\bibitem[Jin19]{Jin19}
Ce~Jin.
\newblock An improved {FPTAS} for 0-1 knapsack.
\newblock In {\em Proceedings of 46th International Colloquium on Automata, Languages, and Programming}, pages 76:1--76:14. Schloss Dagstuhl -- Leibniz-Zentrum f{\"u}r Informatik, 2019.

\bibitem[Jin24]{Jin24}
Ce~Jin.
\newblock 0-1 {{Knapsack}} in {{Nearly Quadratic Time}}.
\newblock In {\em Proceedings of the 56th Annual ACM Symposium on Theory of Computing}, pages 271--282. Association for Computing Machinery, 2024.

\bibitem[JK18]{JK18}
Klaus Jansen and Stefan E.~J. Kraft.
\newblock A faster {{FPTAS}} for the {{Unbounded Knapsack Problem}}.
\newblock {\em European J. Combin.}, 68:148--174, February 2018.
\newblock \href {https://doi.org/10.1016/j.ejc.2017.07.016} {\path{doi:10.1016/j.ejc.2017.07.016}}.

\bibitem[Kar72]{Kar72}
Richard~M. Karp.
\newblock Reducibility among combinatorial problems.
\newblock In Raymond~E. Miller, James~W. Thatcher, and Jean~D. Bohlinger, editors, {\em Complexity of Computer Computations}, pages 85--103. Springer US, Boston, MA, 1972.
\newblock \href {https://doi.org/10.1007/978-1-4684-2001-2_9} {\path{doi:10.1007/978-1-4684-2001-2_9}}.

\bibitem[KP99]{KP99}
Hans Kellerer and Ulrich Pferschy.
\newblock A {{New Fully Polynomial Time Approximation Scheme}} for the {{Knapsack Problem}}.
\newblock {\em J. Comb. Optim.}, 3(1):59--71, July 1999.
\newblock \href {https://doi.org/10.1023/A:1009813105532} {\path{doi:10.1023/A:1009813105532}}.

\bibitem[KP04]{KP04}
Hans Kellerer and Ulrich Pferschy.
\newblock Improved {{Dynamic Programming}} in {{Connection}} with an {{FPTAS}} for the {{Knapsack Problem}}.
\newblock {\em J. Comb. Optim.}, 8(1):5--11, March 2004.
\newblock \href {https://doi.org/10.1023/B:JOCO.0000021934.29833.6b} {\path{doi:10.1023/B:JOCO.0000021934.29833.6b}}.

\bibitem[KPD04]{KPD04}
Hans Kellerer, Ulrich Pferschy, and Pisinger David.
\newblock {\em Knapsack Problems}.
\newblock Springer, Berlin Heidelberg, 2004.
\newblock \href {https://doi.org/10.1007/978-3-540-24777-7} {\path{doi:10.1007/978-3-540-24777-7}}.

\bibitem[KPS17]{KPS17}
Marvin K{\"u}nnemann, Ramamohan Paturi, and Stefan Schneider.
\newblock On the {{Fine-Grained Complexity}} of {{One-Dimensional Dynamic Programming}}.
\newblock In {\em 44th {{International Colloquium}} on {{Automata}}, {{Languages}}, and {{Programming}}}, pages 21:1--21:15. Schloss Dagstuhl -- Leibniz-Zentrum f{\"u}r Informatik, 2017.

\bibitem[Law79]{Law79}
Eugene~L. Lawler.
\newblock Fast {{Approximation Algorithms}} for {{Knapsack Problems}}.
\newblock {\em Math. Oper. Res.}, 4(4):339--356, November 1979.
\newblock URL: \url{https://www.jstor.org/stable/3689221}.

\bibitem[Mao24]{Mao24}
Xiao Mao.
\newblock $(1-\epsilon)$-approximation of knapsack in nearly quadratic time.
\newblock In {\em Proceedings of the 56th Annual ACM Symposium on Theory of Computing}, pages 295--306. Association for Computing Machinery, 2024.

\bibitem[MWW19]{MWW19}
Marcin Mucha, Karol W{\k{e}}grzycki, and Micha{\l} W{\l}odarczyk.
\newblock A {{Subquadratic Approximation Scheme}} for {{Partition}}.
\newblock In {\em Proceedings of the 2019 {{Annual ACM-SIAM Symposium}} on {{Discrete Algorithms}}}, pages 70--88. Society for Industrial and Applied Mathematics, 2019.

\bibitem[PRW21]{PRW21}
Adam Polak, Lars Rohwedder, and Karol W{\k{e}}grzycki.
\newblock Knapsack and {{Subset Sum}} with {{Small Items}}.
\newblock In {\em 48th {{International Colloquium}} on {{Automata}}, {{Languages}}, and {{Programming}}}, pages 106:1--106:19. Schloss Dagstuhl -- Leibniz-Zentrum f{\"u}r Informatik, 2021.

\bibitem[Rhe15]{Rhe15}
Donguk Rhee.
\newblock {\em Faster fully polynomial approximation schemes for knapsack problems}.
\newblock PhD thesis, Massachusetts Institute of Technology, Cambridge, 2015.

\bibitem[Sah75]{Sah75}
Sartaj Sahni.
\newblock Approximate {{Algorithms}} for the 0/1 {{Knapsack Problem}}.
\newblock {\em J. ACM}, 22(1):115--124, January 1975.
\newblock \href {https://doi.org/10.1145/321864.321873} {\path{doi:10.1145/321864.321873}}.

\bibitem[S{\'a}r89]{Sar89}
A.~S{\'a}rk{\"o}zy.
\newblock Finite addition theorems, {{I}}.
\newblock {\em J. Number Theory}, 32(1):114--130, May 1989.
\newblock \href {https://doi.org/10.1016/0022-314X(89)90102-9} {\path{doi:10.1016/0022-314X(89)90102-9}}.

\bibitem[SV06]{SV06}
E.~Szemer{\'e}di and V.~H. Vu.
\newblock Finite and {{Infinite Arithmetic Progressions}} in {{Sumsets}}.
\newblock {\em Ann. of Math.}, 163(1):1--35, January 2006.
\newblock URL: \url{https://www.jstor.org/stable/20159950}.

\end{thebibliography}
\end{document}